\documentclass[onecolumn, 11pt,letterpaper]{article}

\usepackage{graphicx}     			
\usepackage{multirow}               		
\usepackage{amssymb}                	
\usepackage{latexsym}               		
\usepackage{alltt}                  		
\usepackage{amsgen}                 	
\usepackage{amstext}               		
\usepackage{amscd}                  		
\usepackage{algorithmic}
\usepackage{url}
\usepackage{enumerate}
\usepackage{amsthm}
\usepackage{amsmath}
\usepackage{tikz}
\usepackage{float}
\usepackage{enumitem}
\usepackage{mdwlist}
\usepackage{xspace}
\usepackage[left=1in,top=1in,right=1in,bottom = 1in, nohead]{geometry}
\usepackage{bm}                     	
\usepackage{upgreek}                
\usepackage{xspace}            
\usepackage[caption=false]{subfig}            

\usepackage{color}                 
\usepackage{colortbl}             
\usepackage{tabularx}		

\usepackage{aliascnt}  		

\usepackage{float}                 
\usepackage{ifpdf}
    \ifpdf                                                                  
        \usepackage[breaklinks=true]{hyperref}
        \hypersetup{
            colorlinks=false,               
            citebordercolor={0 1 0},        
            urlbordercolor={0 0 1},         
            linkbordercolor={1 0.2 0.2},    
            pdfborder={0 0 1.2},            
        }
    \else
        \usepackage[hypertex, breaklinks=true]{hyperref}       
        \hypersetup{
            colorlinks=true,               
            anchorcolor=yellow,
            linkcolor=green,
            filecolor=green,
            urlcolor=blue,
            citecolor=red,
        }
    \fi
    \hypersetup{                                                            
        pdfpagemode=UseNone,               
        pdfstartview=Fit,
        pdfpagelayout=SinglePage,       
        bookmarks=true,
        bookmarksnumbered=true,
        bookmarksopenlevel=1,            
        bookmarkstype=toc,
        pdftex=true,
        unicode,
        plainpages=false,           
        hypertexnames=true,         
        pdfpagelabels=true,         
        hyperindex=true,
        naturalnames=false,                
    }

\newtheorem{theorem}{Theorem}[section]          	
\newaliascnt{lemma}{theorem}				
\newtheorem{lemma}[lemma]{Lemma}              	
\aliascntresetthe{lemma}  					
\newaliascnt{conjecture}{theorem}			
\aliascntresetthe{conjecture}  				
\newaliascnt{remark}{theorem}				
              
\aliascntresetthe{remark}  					
\newaliascnt{corollary}{theorem}			
\newtheorem{corollary}[corollary]{Corollary}      
\aliascntresetthe{corollary}  				
\newaliascnt{definition}{theorem}			
\newtheorem{definition}[definition]{Definition}    
\aliascntresetthe{definition}  				
\newaliascnt{proposition}{theorem}			
\aliascntresetthe{proposition}  				
\newaliascnt{example}{theorem}			
\newtheorem{example}[example]{Example}  	
\aliascntresetthe{example}  				
  {\refstepcounter{theorem}\trivlist\item       
  [\hskip\labelsep\bf\sf Policy Query \thetheorem]}
  {\endtrivlist}                                
\let\orgautoref\autoref                         		

\renewcommand{\autoref}[1]{
    \def\equationautorefname{Eq.}
    \def\figureautorefname{Fig.}%
    \def\subfigureautorefname{Fig.}%
    \def\lemmaautorefname{Lemma}%
    \def\conjectureautorefname{Conjecture}%
    \def\remarkautorefname{Remark}%
    \def\propositionautorefname{Prop.}%
    \def\corollaryautorefname{Corollary}%
    \def\definitionautorefname{Def.}%
    \def\sectionautorefname{Sect.}%
    \def\subsectionautorefname{Sect.}%
    \def\subsubsectionautorefname{Section}%
    \def\exampleautorefname{Example}%
    \orgautoref{#1}%
}

\newcommand{\specificref}[2]{\hyperref[#2]{#1~\ref*{#2}}}			

\renewcommand{\epsilon}{\varepsilon}    

\renewcommand{\[}{[\![}

\newcommand{\introparagraph}[1]{\textbf{#1.}}        

\definecolor{gray}{rgb}{0.5,0.5,0.5}
\definecolor{niceblue}{rgb}{.8,.85,1}

\newcommand{\set}[1]{\{#1\}}                    
\newcommand{\setof}[2]{\{{#1}\mid{#2}\}}        
\newcommand{\makeop}[2]                         
  {\ifx#2.\def\next##1{}\else\escapechar=-1     
  \def\next##1{\escapechar=92\def#2{#1}}        
  \expandafter\next\expandafter{\string#2}      
  \let\next\makeop\fi\next{#1}}                 
\makeop{{\cal #1}}                              
  \W  .             %

\def \var(#1){{\bf #1}}

\def\AddSpace#1{\ifcat#1a\ \fi#1} 

\newcommand{\silentreminder}[1]{}

\def \up(#1){[#1)}
\def \down(#1){(#1]}

\def \series(#1,#2){#1_1, \dots \; #1_{#2}}
\def \serieszero(#1,#2){#1_0, #1_1, \dots \; #1_{#2}}

\def \para(#1){{\vspace{1ex}\noindent\small\bf #1\hspace{1ex}}}
\def \myem(#1){{\vspace{1ex}\noindent\small\em #1\hspace{1ex}}}

\newcommand{\eat}[1]{}

\renewenvironment{thebibliography}[1]{%
    \begin{oldthebibliography}{#1}%
    	\small
        \setlength{\parskip}{0.1ex}
        \setlength{\itemsep}{0.0ex}%
}%
{%
\end{oldthebibliography}%
}

\setlist[enumerate]{topsep=2pt}

\newcommand{\cut}[1]{}

\newenvironment{packed_item}{
\begin{itemize}
   \setlength{\itemsep}{1pt}
   \setlength{\parskip}{0pt}
   \setlength{\parsep}{0pt}
}
{\end{itemize}}

\newenvironment{packed_enum}{
\begin{enumerate}
   \setlength{\itemsep}{1pt}
  \setlength{\parskip}{0pt}
   \setlength{\parsep}{0pt}
}
{\end{enumerate}}

\newcommand{\dtr}[0]{\twoheadrightarrow}     

\newcommand{\bQ}[0]{\mathbf{Q}}       
\newcommand{\bR}[0]{\mathbf{R}}

\newcommand{\mS}[0]{\mathcal{S}} 
\newcommand{\mI}[0]{\mathcal{I}}   
\newcommand{\mC}[0]{\mathcal{C}}      
   
\newcommand{\mL}[0]{\mathcal{L}}   
\newcommand{\mP}[0]{\mathcal{P}}   

\newcommand{\APS}[0]{\textsf{APS}}   
\newcommand{\DPS}[0]{\textsf{DPS}} 
\newcommand{\QPS}[0]{\textsf{QPS}} 

\newcommand{\agr}[2]{\mS_{#1}(#2)} 
\newcommand{\dagr}[2]{\overline{\mS}_{#1}(#2)} 

\usepackage{aliascnt}

\begin{document}

\title{The Design of Arbitrage-Free Data Pricing Schemes}

\author{Shaleen Deep$^1$ \and Paraschos Koutris$^1$\\
\and$^1$University of Wisconsin-Madison, Madison, WI\\ \{shaleen, paris\}@cs.wisc.edu
}
\date{}
\maketitle

\begin{abstract}
Motivated by a growing market that involves buying and selling data over the web, we study pricing schemes that assign value to queries issued over a database.
Previous work studied pricing mechanisms that compute the price of a query by extending a data seller's explicit prices on certain queries, or investigated the properties that a pricing function should exhibit without detailing a generic construction. 
In this work, we present a formal framework for pricing queries over data that allows the construction of general families of pricing functions, with the main goal of avoiding arbitrage. We consider two types of pricing schemes: instance-independent schemes, where the price depends only on the structure of the query, and answer-dependent schemes, where the price also depends on the query output.  Our main result is a complete characterization of the structure of pricing functions in both settings, by relating it to properties of a function over a lattice. We use our characterization, together with information-theoretic methods, to construct a variety of arbitrage-free pricing functions. Finally, we discuss various tradeoffs in the design space and present techniques for efficient computation of the proposed pricing functions. 
\end{abstract}

\section{Introduction}
\label{sec:introduction}

The commodification of data over the last decade has created many unique research challenges, among them data privacy and pricing of data. In a broad range of application areas, data today is being collected at an unprecedented scale. This phenomenon has led to a growing market for so called big data brokers, who sell this data to buyers such as financial firms, retailers and insurance companies~\cite{cukier2013rise, federal2014data}.

In this paper, we investigate the problem of {\em query-based data pricing}, where the task is to assign prices to queries over a database, such that the price captures the amount of information revealed by asking the query. Traditionally, data pricing has been done either by allowing the buyer to access only certain queries with a fixed price set by the seller, or the buyer needs to purchase the whole dataset~\cite{LK14}. Although such an approach is conceptually simple, defining a large set of queries that are representative of the user's needs is a tall task for the data seller.  Even if this is feasible, such a pricing scheme may allow {\em arbitrage}, which occurs when a data buyer can potentially buy data at a price less than what is set by the seller. It can also lead to prices that exhibit undesirable behavior.

Previous work in the area of data pricing has identified a set of {\em arbitrage conditions} that any reasonable pricing function should avoid. 
The fundamental arbitrage condition is {\em information arbitrage}, first introduced in~\cite{KUBHS12b}. Intuitively, a query $Q_1$ that reveals a subset of the information that is revealed by another query $Q_2$ should be priced at most as much as $Q_2$. If not, an arbitrage opportunity occurs: a clever buyer can pay the price of $Q_2$ and then use the result of $Q_2$ to compute $Q_1$ for a lower price. 
A second arbitrage condition is {\em bundle arbitrage}~\cite{LK14}. Intuitively, asking simultaneously for $Q_1$ and $Q_2$ (as a bundle) should cost at most the sum of asking separately for each. 
Both~\cite{KUBHS12b, LK14} propose pricing functions that avoid both arbitrage conditions. 
However, to the best of our knowledge, there exists no framework that supports a generic construction of pricing functions, and facilitates the analysis of the various tradeoffs in design choices. 
  
\introparagraph{Our Contribution}
We address the question of designing arbitrage-free pricing schemes that assign prices to queries over a database. Our main result is a complete characterization of the structure of pricing functions for two pricing schemes: {\em answer-dependent prices} (\APS), and {\em instance-independent prices} (\QPS). We use this characterization to construct a variety of pricing functions, and also discuss the various tradeoffs involved in choosing the right pricing function. We summarize below our results in more detail.

We first study \APS, where the price depends both on the query $Q$ and on the answer of the query $E = Q(D)$. To characterize such schemes, we define the {\em conflict set}, which is the set of databases such that $Q(D) \neq E$. We show that any arbitrage-free pricing function is equivalent to a monotone and subadditive function over the join-semilattice defined by the conflict sets (Theorems~\ref{thm:arbitrage-free:1} and~\ref{thm:arbitrage-free:2}). Equipped with this characterization, we present several examples of arbitrage-free functions, including the weighted coverage and the weighted set cover functions. In addition, we show that an answer-dependent pricing function with no bundle arbitrage leads to unnatural behavior: any query can cost at least half the price of the whole dataset for some databases. This suggests that there is a tradeoff that any data seller must take into account when choosing a pricing function.


Second, we examine the structure of \QPS, where the pricing function depends only on the query $Q$. We prove that any non-trivial instance-independent pricing function must have weaker arbitrage guarantees compared to an answer-dependent function.
To provide a characterize of functions in \QPS, we view the query $Q$ as a {\em partition} over the set of possible databases: our main results is that any arbitrage-free function is equivalent to a monotone and subadditive function over the elements of the join-semilattice formed from the partitions (Theorems~\ref{thm:arb:1} and~\ref{thm:arb:2}).

To design pricing functions in \QPS, we apply two methods. The first method applies an appropriate aggregate function to combine the prices of an arbitrage-free answer-dependent function (\autoref{lem:average}). The second method views the database as a random variable (with some probability distribution over the possible databases), and computes the price as the {\em information gain} of the data buyer after the answer has been revealed (Section~\ref{sec:entropy}). This approach is parallel to work on side-channel attacks~\cite{kopf2007information}, and quantitative information flow~\cite{Kopf:2010aa}. By using different entropy measures, such as {\em Shannon entropy}, or {\em min-entropy}, we obtain pricing functions that we prove to be arbitrage-free using the machinery we developed.

Third, we show how the proposed pricing functions can be computed efficiently in practical settings. We discuss two different techniques. The first method restricts the computation of a pricing function to a small set of databases (instead of all possible databases). The second method uses approximation techniques to estimate the price within a small margin of error. 

\introparagraph{Organization} 
Section 2 presents the key concepts, terminology and notation that we use throughout the paper. In Section 3, we study the construction and properties of pricing functions for the answer-dependent case. Section 4 details the corresponding problem for instance-independent pricing schemes. Section 5 discusses techniques to compute a pricing function efficiently. We present the related work and conclude in Sections 6 and 7 respectively.


\section{Notation and Framework}
\label{sec:framework}

In this section, we set up the necessary notation and formally describe the pricing framework.

\subsection{Preliminaries}

We fix a relational schema $\bR = (R_1, \dots, R_k)$; we use $D$ to denote a database instance that uses the schema. We will use $\mI$ to denote the set of possible database instances. The set $\mI$ encodes information about the database that is provided by the data seller, and is public information known to any data buyer. Further, we allow the set $\mI$ to be infinite, but countable. For example, suppose that the schema consists of a single binary relation $R(A,B)$ and we know that the domain of both attributes is $[n] = \set{1, \dots, n}$. Then, $\mI = 2^{[n] \times [n]}$, which represents equivalently the set of all possible directed graphs on the vertex set $[n]$.

We will view a {\em query} $Q$ from some query language $\mL$ as a deterministic function that takes as input a database instance $D \in \mI$ and returns an output $Q(D)$. In this paper, we do not impose any restriction on the query language $\mL$, but in the examples we will use and in some of the design tradeoffs we assume $Q$ is either a {\em conjunctive query (CQ)} or a {\em union of conjunctive queries (UCQ)}. 
 A {\em query bundle} $\bQ =  (Q_1, \dots, Q_n )$ is a finite set of queries that is asked simultaneously on the database. We denote by $B(\mL)$ the set of finite query bundles from the language $\mL$. Given two query bundles $\bQ_1, \bQ_2$, we denote their union as $\bQ = \bQ_1, \bQ_2$. 

\introparagraph{Queries as Partitions}
It will be handy to provide an alternative viewpoint of a query bundle $\bQ$ as a  partition over the set of instances $\mI$. A {\em partition} $\mP = \set{B_1, \dots, B_k}$ of $\mI$ is a set of pairwise disjoint sets $B_i \subseteq \mI$, which we call {\em blocks}, such that $\cup_{i=1}^k B_i = \mI$. Given $\bQ \in \mL$, we denote by $\mP_{\bQ}$ the partition that is induced by the following {\em equivalence relation}: $D \sim D'$ iff $\bQ(D) = \bQ(D')$ and $\bQ \in \mL$. In other words, two databases belong in the same block of the partition if and only if their output for $\bQ$ is indistinguishable. 
We use the standard notation $[D]_\bQ$ to denote the equivalence class in which $D$ belongs; in other words, $[D]_\bQ = \setof{D' \in \mI}{ \bQ(D') = \bQ(D)}$. 
For two partitions $\mP_1, \mP_2$, we say that $\mP_1$ {\em refines} $\mP_2$, and write $\mP_1 \succeq \mP_2$, if every block of $\mP_1$ is a subset of some block in $\mP_2$. In other words, $\mP_1$ is a more fine-grained partition of $\mI$ than $\mP_2$. 

\introparagraph{Lattices and Join-Semilattices} A {\em join-semilattice}  $(L, \leq)$ is a partially ordered set in which every two elements in $L$ have a unique supremum (called {\em join} and denoted as $\vee$). A {\em lattice} $(L, \leq)$ is a partially ordered set in which every two elements in $L$ have both a unique supremum, and a unique infimum (called {\em meet} and denoted $\wedge$). In this paper, we will consider two different join-semilattices. The first semilattice has elements subsets of $\mI$, which are ordered by subset inclusion $\subseteq$. The second semilattice has elements partitions of $\mI$, which are ordered by the refinement relation $\preceq$.

Let $f: L \rightarrow \mathbb{R}$ be a function defined on the elements of the join-semilattice. We say that $f$ is {\em monotone}, or {\em isotone}, if  whenever $A \leq B$, then $f(A) \leq  f(B)$. Moreover, we say that $f$ is {\em subadditive} if for any two elements $A,B$ of the semilattice we have $f(A \vee B) \leq f(A) + f(B)$.

\subsection{The Pricing Framework}

In our setting, a data seller offers a database instance $D$ for sale. Data buyers can issue queries on the database in the form of query bundles $\bQ$. For each query $\bQ$ over the instance $D$, the task in hand is to assign a {\em price} to the query answer $\bQ(D)$ that reflects the amount of information gained by the data buyer.
When a price is assigned to a query bundle $\bQ$, we can differentiate between three different pricing strategies, which depend on the parameters used to compute the price.
There are three possible parameters we can use to determine the price of a query: the query bundle $\bQ$, the answer of the query on the database $D$, denoted $E = \bQ(D)$, and the database $D$ itself. The price will obviously depend on which query $\bQ$ we issue, but there is a choice of which $D,E$ should be further used to compute the price. This choice defines three different classes of pricing schemes: 
\begin{packed_item}
\item {\bf Instance-independent (\QPS):} the price depends only on $\bQ$, in which case the pricing function is of the form $p(\bQ)$. The price is independent of the underlying data.
\item {\bf Answer-dependent (\APS):} the price depends on the answer $E = \bQ(D)$, so the price is of the form $p(\bQ, E)$. In this case, the price depends on the query and the query output.
\item {\bf Data-dependent (\DPS):} the price depends on the underlying database $D$, so the pricing function is of the form $p(\bQ, D)$.
\end{packed_item}

Any instance-independent scheme can be cast as an answer-dependent scheme, and any answer-dependent scheme as a data-dependent scheme. The distinction between \APS\ and \DPS\ was introduced in~\cite{LK14}, where the authors use the terminology {\em delayed pricing} and {\em up-front pricing} respectively. Notice that both in \QPS\ and \APS\, the prices themselves do not leak any information about the underlying data $D$.\footnote{For the case of answer-dependent prices, we must make sure that we reveal the price only if we are certain that the buyer will be charged for the cost.} In contrast, a data-dependent pricing scheme can leak information about the data (for more details see~\cite{LK14}). For this reason, in this paper we focus on the first two types of pricing schemes: \QPS\ and \APS. 

The reason we consider query bundles in our setting is that in practice a data buyer will issue over time a sequence $\bQ_1, \dots, \bQ_m$ of query bundles on the database. In this case, after issuing the first $i$ queries, the data buyer should not be charged a price of $\sum_i p(\bQ_i,D)$, but instead $p(\bQ_1, \dots, \bQ_i, D)$. Notice here that, even if a user issues only single queries, we still need to be able to price a query bundle.


\subsection{Arbitrage Conditions}

Assigning prices to query bundles without any restrictions can lead to the occurrence of arbitrage opportunities. In~\cite{KUBHS12}, the authors presented a single condition that captures arbitrage. Here, we follow~\cite{LK14}, and consider independently two different conditions where arbitrage may occur. 

\introparagraph{Information Arbitrage} The first condition captures the intuition that the price of query bundle must capture the amount of information that an answer reveals about the actual database $D$. In particular, if a query bundle $\bQ_1$ reveals a subset of information than a query bundle $\bQ_2$ reveals, the price of $\bQ_1$ must be less than the price of $\bQ_2$. If this condition is not satisfied, it creates an arbitrage opportunity, since a data buyer can purchase $\bQ_2$ instead, and use it to obtain the answer of $\bQ_1$ for a cheaper price. 

\introparagraph{Bundle Arbitrage} The second condition regards the scenario where a data buyer that wants to obtain the answer for the bundle $\bQ = \bQ_1, \bQ_2$ creates two separate accounts, and uses one to ask for $\bQ_1$ and the other to ask for $\bQ_2$. To avoid such an arbitrage situation, we must make sure that the price of $\bQ$ is at most the sum of the prices for $\bQ_1$ and  $\bQ_2$. \cite{LK14} uses the terminology {\em separate-account arbitrage} to refer to this arbitrage condition.

We will show in the next sections how to mathematically formalize information arbitrage and bundle arbitrage for both \APS\ and \QPS.

\section{Answer-Dependent Pricing}
\label{sec:aps}

In this section, we study the design of {\em answer-dependent} pricing schemes. In an \APS\ the pricing function takes the form $p(\bQ, E)$, where $\bQ$ is a query bundle and $E \in \setof{\bQ(D)}{D \in \mI}$. Throughout the section, we assume that query bundles belong to some query language $\mL$.
We first discuss how to formalize the arbitrage conditions. To formally describe information arbitrage, we use the notion of {\em data-dependent determinacy}.

\begin{definition}
We say that $\bQ_2$ determines $\bQ_1$ under database $D$, denoted $D \vdash \bQ_2 \dtr \bQ_1$ if for every database $D'$ such that $\bQ_2(D) = \bQ_2(D')$, we also have  $\bQ_1(D') = \bQ_1(D)$.
\end{definition}

The above definition of determinacy is different from {\em query determinacy}~\cite{NSV07, NSV10}, since it is defined with respect to a given database $D$. It is also easy to see that if $D \vdash \bQ_2 \dtr \bQ_1$, we also have that $D' \vdash \bQ_2 \dtr \bQ_1$ for any database $D'$ such that $\bQ_2(D) = \bQ(D')$.

\begin{definition}[\APS\ Information Arbitrage]
Let $\bQ_1, \bQ_2$ be two query bundles. We say that the pricing function $p$ has no {\em information arbitrage} if for every database $D \in \mI$, $D \vdash \bQ_2 \dtr \bQ_1$ implies that $p(\bQ_2, E_2) \geq  p(\bQ_1, E_1)$, where $E_i = \bQ_i(D)$ for $i=1,2$.
\end{definition}

This definition of information arbitrage captures both {\em post-processing arbitrage} and {\em serendipitous arbitrage}, as these are defined in~\cite{LK14}. 
For the case of bundle arbitrage, we formalize it as follows.

\begin{definition}[\APS\ Bundle arbitrage]
Let the query bundle $\bQ = \bQ_1, \bQ_2$. We say that the price function $p$ has no {\em bundle arbitrage} if for every database $D \in \mI$, we have 
$p(\bQ, E) \leq p(\bQ_1, E_1) + p(\bQ_2, E_2)$,
where $E = \bQ(D)$ and $E_i = \bQ_i(D)$ for $i=1,2$.
\end{definition}

We say that an answer-dependent pricing function is {\em arbitrage-free} if it has no information arbitrage and no bundle arbitrage.

\subsection{How to Find a Pricing Function}

In this section, we characterize the family of answer-dependent pricing functions that satisfy both arbitrage conditions. The critical component is the notion of a {\em conflict set}.

\subsubsection{Conflict Sets}

Consider a query bundle $\bQ \in B(\mL)$, a database $D \in \mI$ and let $E = \bQ(D)$. We define%
\begin{align*}
\agr{\bQ}{E}  = \setof{D' \in \mathcal{I}}{\bQ(D') = E}, \quad \quad \quad 
\dagr{\bQ}{E}  = \setof{D' \in \mathcal{I}}{\bQ(D') \neq E} 
\end{align*}
In other words, $\agr{\bQ}{E}$ computes the set of databases that ``agree'' with the view extension $E$, and $\dagr{\bQ}{E}$ contains the complement set, i.e. the set of databases that ``disagree'' with $E$. Notice that $\agr{\bQ}{\bQ(D)} = [D]_\bQ$. We refer to $\dagr{\bQ}{E}$ as the {\em conflict set} for query $\bQ$ and extension $E$, while we refer to $\agr{\bQ}{E}$ as the {\em agreement set}. It is straightforward that $\dagr{\bQ}{E} = \mI \setminus \agr{\bQ}{E}$. 


\begin{example} \label{ex:intro}
We will use the following scenario as a running example throughout this section. Suppose that we have a binary relation $R(\underline{A},B)$, where attribute $A$ is the key. The values of the $n$ keys are also publicly known $\set{a_1, a_2, \dots, a_n}$. Moreover, assume that $B$ can take two possible values from $\set{0,1}$. It is easy to see that $\mI$ consists of $2^n$ databases. For $n=2$, let $D_{ij}$ denote the database $\{(a_1,i), (a_2,j)\}$. For example $D_{01} = \{(a_1,0), (a_2,1)\}$.

Consider now the query $Q(x) = R(a_1,x)$, which asks for value of attribute $B$ for the tuple with key $A=a_1$. Assume that the underlying database is $D_{01}$. The conflict set of $Q$ and $E = Q(D_{01})$ consists of all databases $D$ for which $(a_1,1) \in D$, hence
$\agr{Q}{E} = \set{D_{10}, D_{11}}$.
\end{example}

If  $\bQ$ returns a constant answer for every database in $\mI$, the conflict set will be the empty set. On the other hand, if $\bQ$ reveals the whole database $D$, the conflict set will be $\mI \setminus \set{D}$. We can now define the set of all possible conflict sets for  a database $D$ and a given language $\mL$ as $\mS_{D}^\mL = \{\dagr{\bQ}{\bQ(D)} \mid \bQ \in B(\mL)\}$. The following lemma shows that $\mS_D^\mL$ forms a {\em join-semilattice} under the partial order $\subseteq$, where the join operator is set union. 

\begin{lemma} \label{lem:bundle:char}
Let $\bQ = \bQ_1, \bQ_2$. For a database $D \in \mI$, let $E_1 = \bQ_1(D)$, $E_2 = \bQ_2(D)$, and $E = \bQ(D)$. Then,
$ \dagr{\bQ}{E} = \dagr{\bQ_1}{E_1} \cup \dagr{\bQ_2}{E_2}.$
\end{lemma}

\begin{proof}
Let us denote $A_1 =  \agr{\bQ_1}{E_1}$, $A_2 =  \agr{\bQ_2}{E_2}$ and $A = \agr{\bQ}{E}$. It is easy to see that $A = A_1 \cap A_2$, since by definition $A$ contains exactly the databases that agree with respect to both $\bQ_1$ and $\bQ_2$. Taking complements we obtain that  $\overline A = \overline{ A_1} \cup \overline{A_2}$. Also, notice that since both $\bQ_1,\bQ_2 \in B(\mL)$, $\bQ \in B(\mL)$.
\end{proof}

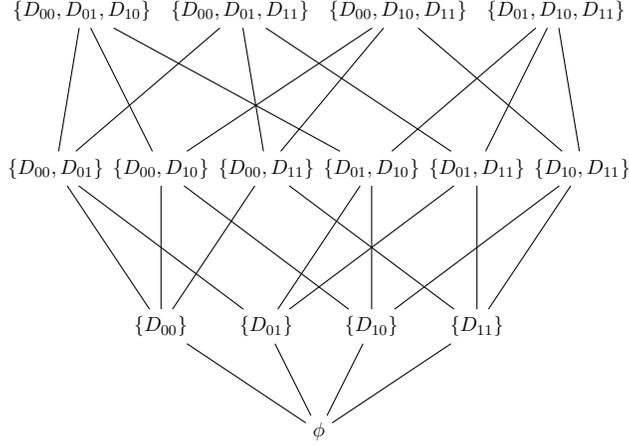
\begin{figure}
\centering
\begin{tikzpicture}[scale=0.7, transform shape]
  \node (0) at (-4.5,6) {$\set{ D_{00} ,D_{01},D_{10}}$};
  \node (1) at (-1.5,6) {$\set{ D_{00} ,D_{01},D_{11}}$};
  \node (2) at (1.5,6) {$\set{ D_{00} ,D_{10},D_{11}}$};
  \node (3) at (4.5,6) {$\set{ D_{01} ,D_{10},D_{11}}$};
  \node (e) at (-5,3) {$\{ D_{00},D_{01}\}$};
  \node (f) at (-3,3) {$\{ D_{00},D_{10}\}$};
  \node (g) at (-1,3) {$\{ D_{00},D_{11}\}$};
  \node (h) at (1,3) {$\{ D_{01},D_{10}\}$};
  \node (i) at (3,3) {$\{ D_{01},D_{11}\}$};
  \node (j) at (5,3) {$\{ D_{10},D_{11}\}$};
  \node (a) at (-3,0) {$\{ D_{00} \}$};
  \node (b) at (-1,0) {$\{ D_{01} \}$};
  \node (c) at (1,0) {$\{ D_{10} \}$};
  \node (d) at (3,0) {$\{ D_{11} \}$};
  \node (min) at (0,-2) {$\phi$};
  \draw 
  (a) -- (min) -- (b)  (c) -- (min) -- (d)
  (a) -- (e) -- (b) -- (h) -- (c) -- (j) -- (d)
  (e) -- (0) -- (f) -- (2) -- (j) -- (3) ;
  \draw
  (a) -- (f) -- (c)
  (b) -- (i) -- (d)  
  (a) -- (g) -- (d);
  \draw
  (3) -- (i) -- (1) -- (g) --(2)  
  (e) -- (1) 
  (0) -- (h) -- (3) ;
\end{tikzpicture}
\caption{A simultaneous depiction of the join-semilattices for the four databases in Example~\ref{ex:intro}.}
\label{fig:incl:lattice}
\end{figure}

The diagram in Figure~\ref{fig:incl:lattice} depicts simultaneously the four join-semilattices for each of the databases in Example~\ref{ex:intro}. We next prove a lemma that connects the notion of a conflict set with data-dependent determinacy.

\begin{lemma} \label{lem:equiv:dtr}
Let $\bQ_1, \bQ_2$ be two query bundles, and $D \in \mI$ be a database. Let $E_i = \bQ_i(D)$ for $i=1,2$.
The following two statements are equivalent:
\begin{packed_enum}
\item $D \vdash \bQ_2 \dtr \bQ_1$
\item $\dagr{\bQ_2}{E_2} \supseteq \dagr{\bQ_1}{E_1}$
\end{packed_enum}
\end{lemma}

\begin{proof}
$1 \implies 2$. Consider a database $D' \in \agr{\bQ_2}{E_2} $. By definition, it must be that $\bQ_2(D') = E_2 = \bQ_2(D)$. By the definition of data-dependent determinacy, this implies that $\bQ_1(D') = \bQ_1(D) = E_1$, and thus $D' \in \agr{\bQ_1}{E_1}$. This implies in turn that $\agr{\bQ_2}{E_2} \subseteq \agr{\bQ_1}{E_1}$. Taking the complement, we obtain $\dagr{\bQ_2}{E_2} \supseteq \dagr{\bQ_1}{E_1}$.  

 $2 \implies 1$. Consider a database $D'$ such that $\bQ_2(D') = E_2 = \bQ_2(D)$. Then, by definition $D' \in \agr{\bQ_2}{E_2}$, which implies that $D' \in \agr{\bQ_1}{E_1}$. But then we have that $\bQ_1(D') = E_1 = \bQ_1(D)$. Thus, $D \vdash \bQ_2 \dtr \bQ_1$. 
\end{proof} 


Lemma~\ref{lem:equiv:dtr} and Lemma~\ref{lem:bundle:char} demonstrate that information and bundle arbitrage can be cast as conditions on the elements of the semilattice of conflict sets.

\begin{example}
Continuing Example~\ref{ex:intro}, consider the queries $Q_1(x) = R(a_1,x)$ and $Q_2() = R(x,1)$. Let $D_{00}$ be the underlying database. It is easy to see that $D_{00} \vdash Q_2 \dtr Q_1 $, since after asking $Q_2$ we learn that the database contains no 1 values for $B$, and thus it must have only 0 values. The conflict sets for $E_1 = Q_1(D_{00})$, $E_2 = Q_2(D_{00})$ are $\dagr{Q_1}{E_1} = \set{D_{11}, D_{10}}$ and $\dagr{Q_2}{E_2} = \set{D_{01}, D_{10}, D_{11}}$ respectively. 
\end{example}

%

\subsubsection{A Characterization of Arbitrage-Free \APS\ }

We can now use the notion of a conflict set to define  pricing functions of the form
$ p(\bQ, E) = f(\dagr{\bQ}{E})$,
where $f: 2^{\mathcal{I}} \setminus \{I\} \rightarrow  \mathbb{R}_+$ is a {\em set function}. It is straightforward to see that such a pricing function is by construction in \APS, since the computation depends only on $\bQ$ and $E$, and not on the database $D$.
For example, if $\bQ$ returns a constant answer for every database in $\mI$, $p(\bQ, E) = f(\emptyset)$. On the other hand, if $\bQ$ reveals the whole database $D$, $p(\bQ,E) = f(\mI \setminus \set{D})$. 
We can now show a necessary and sufficient characterization of answer-dependent functions with no information arbitrage in terms of such a function $f$.

\begin{theorem}
\label{thm:arbitrage-free:1}
Let $p$ be an answer-dependent pricing function. The following two statements are equivalent:
\begin{packed_enum}
\item $p$ has no information arbitrage. 
\item $p(\bQ, E) = f(\dagr{\bQ}{E})$, where $f$ is a monotone function over every semilattice $\mS_D^\mL$.
\end{packed_enum} 
\end{theorem}

We next present the proof of~\autoref{thm:arbitrage-free:1} using two lemmas, one for each direction of the equivalence. 

\begin{lemma} \label{lem:help:conflict}
Let $ p(\bQ, E) = f(\dagr{\bQ}{E})$ be a pricing function.  If $f$ is a monotone function over every semilattice $\mS_D^\mL$, then $p$ has no information arbitrage.
\end{lemma}

\begin{proof}
Consider two query bundles $\bQ_1, \bQ_2 \in B(\mL)$ such that $D \vdash \bQ_2 \dtr \bQ_1$. From Lemma~\ref{lem:equiv:dtr} this implies that $\dagr{\bQ_1}{E_1}\subseteq \dagr{\bQ_2}{E_2}$, where $E_i = \bQ_i(D)$ for $i=1,2$. Since $f$ is monotone on the semilattice $\mS_D^\mL$, we have
$$ p(\bQ_1, E_1) = f(\dagr{\bQ_1}{E_1} ) \leq f(\dagr{\bQ_2}{E_2} ) =  p(\bQ_2, E_2).$$
This completes the proof.
\end{proof}

\begin{lemma} \label{lem:help:rev}
Let $p$ be an answer-dependent pricing function with no information arbitrage. Then, $p$ is of the form $p(\bQ, E) = f(\overline \mS_{\bQ}(E))$, where $f$ is a monotone function over every semilattice $\mS_D^\mL$.
\end{lemma}

\begin{proof}
We prove this lemma in two steps. In the first step, we prove that $p(\bQ, E) = g( \mS_{\bQ}(E))$ for some function $g$. To prove this statement, we will show that for any two queries $\bQ_1, \bQ_2$,  $\mS_{\bQ_1}(E_1) = \mS_{\bQ_2}(E_2)$ implies that they have the same price. From the fact that $\mS_{\bQ_2}(E_2) \subseteq \mS_{\bQ_1}(E_1)$ and Lemma~\ref{lem:equiv:dtr} we obtain that $D \vdash \bQ_2 \dtr \bQ_1$. Since $p$ has no information arbitrage, it must be that $p(\bQ_2, E_2) \geq p(\bQ_1, E_1)$. Using a symmetric argument, we can also prove that $p(\bQ_2, E_2) \leq p(\bQ_1, E_1)$, which implies that the prices are indeed the same: $p(\bQ_2, E_2) = p(\bQ_1, E_1)$. This proves the existence of such a function $g$.

Define now the function $f(\mS) = g(\mathcal{I} \setminus \mS)$ for every $\mS \subseteq \mathcal{I}$. Then we can write
$$ p(\bQ, E) = g( \mS_{\bQ}(E)) = f(\mathcal{I} \setminus \mS_{\bQ}(E)) =  f(\dagr{\bQ}{E}).$$

In the second step, we will prove the monotonicity of the function $f$ on every semilattice $\mS_D^\mL$ for $D \in \mI$. Suppose $A \subseteq B$, where $A,B \in \mS_D^\mL$. By the definition of $\mS_D^\mL$, there exist $\bQ_A, \bQ_B \in B(\mL)$ such that $\dagr{\bQ_A}{\bQ_A(D)} = A$ and $\dagr{\bQ_B}{\bQ_B(D)} = B$.
%
%
Notice now that since $A \subseteq B$, we have $\overline \mS_{\bQ_A}(E_A)  \subseteq \overline \mS_{\bQ_B}(E_B) $, which by Lemma~\ref{lem:equiv:dtr} implies $D \vdash \bQ_B \dtr \bQ_A$. Since $p$ has no information arbitrage, $f(A) =  f(\overline \mS_{\bQ_A}(E_A)) =  p(\bQ_A, E_A) \leq p(\bQ_B, E_B) = f(B)$. We have thus shown that $A \subseteq B$ implies $f(A) \leq f(B)$.
\end{proof}

We have shown that in order to avoid information arbitrage it suffices to restrict the function to be monotone. We next demonstrate a similar connection of bundle arbitrage to the property of subadditivity. 

\begin{theorem}
\label{thm:arbitrage-free:2}
Let $p(\bQ,E) = f(\dagr{\bQ}{E})$ be a pricing function, where $f$ is a set function. Then, the following two statements are equivalent:
\begin{packed_enum}
\item $p$ has no bundle arbitrage. 
\item $f$ is subadditive over every semilattice $\mS_D^\mL$.
\end{packed_enum} 
\end{theorem}

\begin{proof}
For the direction $2 \implies 1$, fix some database $D \in \mI$. Suppose that $p(\bQ, E) = f(\dagr{\bQ}{E})$ and $f$ is subadditive over $\mS_D^\mL$. Consider the bundle $\bQ = \bQ_1, \bQ_2$. Let $A_1 =  \mS_{\bQ_1}(E_1)$, $A_2 =  \mS_{\bQ_2}(E_2)$ and $A = \mS_{\bQ}(E)$. Applying Lemma~\ref{lem:bundle:char}, we have that $\overline A = \overline{ A_1} \cup \overline{A_2}$. Since $f$ is a subadditive function:
\begin{align*}
p(\bQ, E)  = f(\overline A) \leq f(\overline{A_1}) + f(\overline{A_2}) 
=  p(\bQ_1, E_1)  + p(\bQ_2, E_2) 
\end{align*}

For the direction $1 \implies 2$, let $A_1,A_2 \in \mS_D^\mL$. By the definition of the
semilattice, there exist query bundles $\bQ_1, \bQ_2 \in B(\mL)$ such that the conflict sets are $A_1, A_2$ respectively. Let $\bQ = \bQ_1, \bQ_2$, and notice that $\bQ \in B(\mL)$. Then: 
\begin{align*}
f(A_1 \cup A_2) & = f(\overline \mS_{\bQ_1}(E_1) \cup \overline \mS_{\bQ_2}(E_2) )  = f(\overline \mS_{\bQ}(E) ) \\
& = p(\bQ, E)  \leq p(\bQ_1, E_1)  + p(\bQ_2, E_2)  \\
& =  f(\overline \mS_{\bQ_1}(E_1)) +  f(\overline \mS_{\bQ_2}(E_2)) \\
&  = f(A_1) + f(A_2)
\end{align*}
This concludes the equivalence proof.
\end{proof}

Observe that if a function $f$ is monotone and subadditive over $2^\mI$, it will also be monotone and subadditive over every semilattice $\mS_D^\mL$. Hence, as a corollary we can describe a general family of arbitrage-free pricing functions.

\begin{corollary} \label{cor:arbitrage}
Let $f$ be a monotone and subadditive set function $f$. Then, the function $p(\bQ, E) = f(\overline \mS_{\bQ}(E))$ is an answer-dependent pricing function that is arbitrage-free.
\end{corollary}

\subsection{Explicit Constructions of Pricing Functions}

We have so far described a general class of functions that are both information and bundle arbitrage-free. Since any submodular function is also subadditive, any monotone submodular set function $f$ will also produce a desired pricing function. We give some concrete examples of arbitrage-free pricing functions below.

\begin{corollary}
Suppose that we assign a weight of $w_D$ to each $D \in \mathcal{I}$, such that  $\sum_{D \in \mathcal{I}} w_D < \infty$. Then, the following pricing functions are arbitrage-free:
\begin{packed_enum}
\item the weighted coverage function: $\sum_{D: \bQ(D) \neq E} w_{D}$.
\item the supremum function: $ \sup_{D: \bQ(D) \neq E} w_D$.\footnote{The supremum becomes equivalent to the $\max$ function if $\mI$ is finite.}
\item the budget-limited weighted coverage function for some  $B \geq 0$:
$\min \{ B, \sum_{D: \bQ(D) \neq E} w_D \}$.
\end{packed_enum}
\end{corollary}

We can construct richer pricing functions by combining the weighted coverage function
with a concave function $g$. Indeed, we can show that  $p(\bQ, E) = g(\sum_{D \in \dagr{\bQ}{E}} w_D)$ is arbitrage-free for any concave function $g$. If $\mI$ is finite, we can assign to each database $D \in \mathcal{I}$ an equal weight, in which case we obtain the arbitrage-free function
$p(\bQ, E) = g(|\dagr{\bQ}{E}|)$.





\begin{corollary} \label{corollary:concave}
Suppose that we assign a weight of $w_D$ to each $D \in \mathcal{I}$, such that  $\sum_{D \in \mathcal{I}} w_D < \infty$.  Then, the pricing function $p(\bQ, E) = g(\sum_{D \in \dagr{\bQ}{E}} w_D)$ is arbitrage-free for any concave function $g$.
\end{corollary}

\begin{proof}
We know that if $f(A)$ is a modular set function and $g$ is concave, then $g(f(A))$ is a submodular function. Notice that $f(A) = \sum_{i \in A} w_i$ is a modular function for any
choice of weights $w_i$.
\end{proof}

The pricing functions we have presented thus far are constructed by assigning a weight to each database in $\mI$. Another type of construction starts by specifying a family $\mathcal{F}$ of subsets of $\mI$. For each subset $S \in \mathcal{F}$, we assign a weight $w_S$. Finally, we pick some real number $B \geq \max_{S \in \mathcal{F}} w_S$. We define the  {\em weighted set cover} function $f(A)$ as the cost of the minimum set cover for $A$ if such a set exists, otherwise $f(A) = B$.

\begin{lemma}\label{lem:aps:weighted}
The weighted set cover pricing function is arbitrage-free.
\end{lemma}

\begin{proof}
From~\autoref{cor:arbitrage}, it suffices to show that the set cover function is monotone and subadditive. Indeed, let $A_1 \subseteq A_2$. If $A_1$ is minimally covered by a subset $F \subseteq \mathcal{F}$, this subset also covers $A_2$, so the covering cost for $A_2$ will be at most that of $A_1$. If $A_1$ can not be covered, $A_2$ will also not be covered, so they both have value $B$. 

For subadditivity, let $A = A_1 \cup A_2$. Let $F_1, F_2$ be the minimum covers for $A_1, A_2$ respectively. Then, $F_1 \cup F_2$ is a cover for $A$ with cost at most $f(A_1) + f(A_2)$ (since some sets may overlap). If $A_1$ is not covered, then $f(A_1) +f(A_2) \geq B \geq f(A)$, since $B$ is always greater than the maximum weight.
\end{proof}

The weighted set cover function generalizes the approach from~\cite{KUBHS12}, where explicit prices are specified for certain views, and the price of the query is computed as the cheapest set of views that determine the query.  Indeed, if we are given explicit price points $(\bQ_i, p_i)$ for $i=1, \dots, m$, we can define the following family of sets: $\mathcal{F} = \setof{\dagr{\bQ_i}{\bQ_i(D)}}{i=1, \dots, m}$, where each set $\dagr{\bQ_i}{\bQ_i(D)}$ is assigned a weight of $p_i$. Since $D \vdash \bQ_{i_1}, \dots, \bQ_{i_\ell} \dtr \bQ$ is equivalent to saying that the union of the conflict sets of $\bQ_{i_1}, \dots, \bQ_{i_\ell}$ is a superset of the conflict set of $\bQ$, the minimum set cover for $\dagr{\bQ}{E}$ corresponds to the cheapest set of views that determine $\bQ$ under database $D$.

\subsubsection{Information Gain as a Pricing Function}

A natural mechanism for pricing is to start from a probabilistic point of view and compute the price as the reduction in uncertainty, or {\em information gain}, using some notion of entropy. 

Formally, consider an initial probability distribution over the set $\mathcal{I} $ of possible databases: in other words, assign a probability $p_D$ to each database $D \in \mathcal{I}$. This probability distribution may reflect public information about the database (for example some value might be more probable than some other value). 
Let $X$ be a random variable such that $P(X = D) = p_D$. Given some entropy measure $H(\cdot)$ of a random variable, such as Shannon entropy or min-entropy, we can set the price as the {\em information gain}: the initial entropy $H(X)$ minus the entropy of the new distribution, which is now conditioned on the event $\bQ(X)=E$. Formally, 
%
we define the price as $p(\bQ, E) = H(X) - H(X \mid \bQ(X) = E)$.
%
We can now plug standard uncertainty measures to obtain a pricing function. For example, we can use the Shannon entropy $ H(X) = - \sum_{D \in \mathcal{I}} p_D \log(p_D)$, or the min-entropy $H_\infty(X) = -\log(\max_D p_D)$.

\begin{lemma} \label{lem:counterexample}
There exists a probability distribution $p_D$ over $\mathcal{I}$ such that the answer-dependent entropy function has information-arbitrage.
\end{lemma}

\begin{proof}
Consider two sets $B \subseteq A \subseteq \mathcal{I}$, such that $A \setminus B = \{D_0\}$. Assume that the probabilities are set as follows: for every $D \in B$ we have $p_D = \epsilon$, and $p_{D_0} = 1 - m\epsilon$, where $m = |A|$. Define now two queries $\bQ_{A}$ and $\bQ_{B}$ such that $\mS_{\bQ_{A}}(E) = A$ and $\mS_{\bQ_{B}}(E) = B$. In this case, we have:
\begin{align*}
p(\bQ_{B}, E) & = H(D) + \sum_{i=1}^m \frac{1}{m} \log(1/m)  = H(D) - \log(m) \\
p(\bQ_{A}, E) & = H(D) + m \epsilon \log(\epsilon) + (1-m\epsilon)\log(1-m\epsilon)
\end{align*}
Further, $0 < m\epsilon < 1$. To create a counterexample, we choose $m\epsilon = \frac{1}{2}$, and now we have:
\begin{align*}
p(\bQ_{A}, E) & - p(\bQ_{B}, E) = \\
 & =m \epsilon \log(\epsilon) + (1-m\epsilon)\log(1-m\epsilon) + \log(m) \\
 & = \frac{1}{2} \log(\epsilon) - \frac{1}{2} + \log(m) = \frac{1}{2} \log(m)-1
\end{align*}
By picking $m$ large enough, we can make this quantity strictly positive, hence violating the information arbitrage condition.
\end{proof}

The intuition in the above proof is the following: the result for query $\bQ_A$ will have a somewhat small entropy, because $D_0$ is much more probable than the other databases. However, by asking $\bQ_B$ we learn that $D_0$ cannot be the actual database, and now the probability is equally distributed among the rest of the candidates; hence, the entropy grows! 

The information gain, even though it seems a natural candidate, is not a well-behaved pricing function for \APS, since it exhibits both information and bundle arbitrage (see Lemma~\ref{lem:counterexample} for such an example of information arbitrage). As we will see in Section~\ref{sec:framework_instance_independent} though, we can use information gain to construct arbitrage-free functions for \QPS.
In the case where the probabilities $p_D$ are all equal, the information gain based on Shannon entropy has no information arbitrage (but can still exhibit bundle arbitrage).

\begin{lemma} \label{arbitrage:uniform}
If the probability distribution $p_D$ over $\mathcal{I}$ is uniform, the information gain based on Shannon entropy has no information arbitrage.
\end{lemma}

\begin{proof}
Let $n = |\mI|$. Then, the pricing function can be written as $ p(\bQ,E) = \log(n) - \log(|\mS_{\bQ}(E)|) = \log \left( \frac{n }{n - |\dagr{\bQ}{E}|} \right)$, which is a monotone set function on  $\dagr{\bQ}{E}$.
\end{proof}

\subsection{A Tradeoff for Arbitrage-Free \APS}


\begin{example}
Continuing Example~\ref{ex:intro}, consider the query $Q(x) = R(a,x)$ and the pricing function
$p_2(\bQ, E) = \log(|\dagr{\bQ}{E}|)$. Notice that, independent of the actual database $D$, the conflict set has always size $2^{n-1}$. In this case, $p_2(Q,E) = n-1$. Notice that the price  for learning the whole database is $\log(2^n-1)$, which means that for learning a single tuple we pay almost as much as the whole database.
\end{example}

We will show here that the above example is not a random occurrence, and that the requirement that a pricing function has no bundle arbitrage gives rise to the phenomenon of assigning high prices (w.r.t. to the price of the whole dataset) to queries that reveal only a small amount of information.

\begin{lemma} \label{lem:tradeoff:1}
Let $p(\bQ, E) = f(\dagr{\bQ}{E})$ be an answer-dependent pricing function where $f$ is monotone and subadditive over $2^{\mI}$. Then, for every non-constant query $\bQ \in B(\mL)$ there exists a database $D \in \mI$ such that $p(\bQ, \bQ(D))$ is at least half the price of  $D$.
\end{lemma}

\begin{proof}
Consider a query bundle $\bQ$. Since $\bQ$ is not constant, we can find two databases $D_0$ and $D_1$ such that $\bQ(D_0) \neq \bQ(D_1)$. Let $\mathcal{I}_0 = \overline \mS_{\bQ}(E_0)$ and $\mathcal{I}_1 = \overline \mS_{\bQ}(E_1)$ denote the conflict set for the query $\bQ$ for the cases when $D_0$ and $D_1$ is the actual database respectively. Note that $\overline \mS_{\bQ}(E_0) = \{D' | \bQ(D_0) \neq \bQ(D')\}$ and $\overline \mS_{\bQ}(E_1) = \{D'' | \bQ(D_1) \neq \bQ(D'')\}$. Clearly, $D_0, D_1 \notin \mathcal{I}_0 \cap \mathcal{I}_1$. For every other database $D'$, we note that it belongs to $\mathcal{I}_0 \cup \mathcal{I}_1 = \mathcal{I}$.

Let $p(\bQ,E)$ denote the pricing function. Since $p$ is arbitrage-free, it can be written in the form $p(\bQ, E) = f(\dagr{\bQ}{E})$, where $f$ is subadditive and monotone. Then we have:
%
\begin{align*}
p(\bQ, E_1) + p(\bQ, E_0) & = f(\mI_1) + f(\mI_0) 
 \geq  f(\mI_1 \setminus \{ D_1\}) + f(\mI_0) \nonumber 
 \geq f(\mI - \{ D_1\})\nonumber
\end{align*}
where the first inequality comes from monotonicity, and the second from subadditivity.  We should note here that the requirement that $f$ is monotone and subadditive over all possible subsets of $\mI$ is crucial for the proof. The value $f(\mI - \{ D_1\})$ is equal to the price of the database $D_1$, which we denote by $r_1$. We can similarly show that $p(\bQ, E_1) + p(\bQ, E_0) \geq f(\mI - \{ D_0\}) = r_0$. Summing the two inequalities, we obtain $p(\bQ, E_1) + p(\bQ, E_0) \geq (r_0+r_1)/2$. This implies in turn that either $p(\bQ, E_0) \geq r_0/2$, or $p(\bQ, E_1) \geq r_1/2$. In other words, for either $D_0$ or $D_1$, the price of query $\bQ$ is at least half the price of the whole database, independent of the number of tuples in the database.
\end{proof}



To see that the bundle-arbitrage requirement cause the problem, consider the function $p(\bQ,E) = \log(|\mI|) - \log(|\agr{\bQ}{E}|)$, for which we showed that it exhibits no information arbitrage, but can still have bundle arbitrage. Continuing our example, we can see that $p(Q,E) = \log(2^n) - \log(2^{n-1}) = 1$; thus, learning about one of the $n$ tuples is priced reasonably to $1/n$ of the price of the whole database. Our analysis demonstrates an important tradeoff in the design space of answer-dependent pricing functions: {\em ensuring no bundle arbitrage implies that the pricing function will charge disproportionately high prices for little information}. 

It is also instructive to note that while Lemma~\ref{lem:tradeoff:1} guarantees that existence of database $D \in \mathcal{I}$ that behaves badly, it does not say anything about the number of such databases. In fact, for our example we can show that for query $Q$ at least half of the databases in $\mI$ will exhibit this undesirable behavior.

\section{Instance-Independent Pricing}
\label{sec:framework_instance_independent}

We study here the structure of {\em instance-independent} pricing schemes. In a \QPS, the pricing function is of the form $p(\bQ)$, depending only on the query.  We first formalize the conditions under which the pricing function has no information arbitrage and no bundle arbitrage. 

\begin{definition}
We say that $\bQ_2$ determines $\bQ_1$, denoted $\bQ_2 \dtr \bQ_1$, if for every database database $D'$ and $D''$, $\bQ_2(D') = \bQ_2(D'')$ implies  $\bQ_1(D') = \bQ_1(D'')$.
\end{definition}

In contrast to answer-dependent pricing functions, where we used a notion of determinacy that depends on the database, here we use the standard notion of {\em information-theoretic determinacy}.\footnote{Here we should note that there exists a slight difference, since the databases we consider can come only from $\mI$, and not be any database.} We can now describe the formal definition for information arbitrage.

\begin{definition}[\QPS\ Information Arbitrage]
The pricing function $p$ has no {\em information arbitrage} if for any two query bundles  
$\bQ_1, \bQ_2$ such that $\bQ_2 \dtr \bQ_1$, we have $ p(\bQ_2) \geq  p(\bQ_1)$.
\end{definition}

\begin{definition}[\QPS\ Bundle arbitrage]
Let the query bundle $\bQ = \bQ_1, \bQ_2$. We say that the pricing function $p$ has no {\em bundle arbitrage} if we have $p(\bQ) \leq p(\bQ_1) + p(\bQ_2)$.
\end{definition}

\subsection{Serendipitous Arbitrage}

Consider two query bundles $\bQ_1$ and $\bQ_2$ such that $\bQ_1 \not \dtr \bQ_2$, but for some $D \in \mI$, $D \vdash \bQ_1 \dtr \bQ_2$. For example, consider the boolean query $Q_1() = R(x,y)$ over the binary relation $R(A,B)$. Let $Q_2(x,y) = R(x,y)$. Clearly, for all databases $D$ other than the empty database, $D \vdash Q_1 \not \dtr Q_2$.  However, for the database $D_0 = \emptyset$, note that $D_0 \vdash Q_1 \dtr Q_2$. In this case, if $p(Q_1) > p(Q_2)$, the data buyer would have an arbitrage opportunity. However, this opportunity would arise by chance, since the buyer does not know the underlying database and thus does not know that asking for $\bQ_2$ can lead to learning $\bQ_1$ for a lower price. We call this phenomenon \emph{serendipitous arbitrage}~\cite{LK14}.
Our definition of \QPS\ information arbitrage does not capture serendipitous arbitrage. The next result demonstrates a second tradeoff in the design space of pricing functions: {\em any non-trivial \QPS\ will exhibit serendipitous arbitrage}. 

\begin{theorem} \label{qps:arbitrage}
Let $\mL = UCQ$. If a \QPS\ exhibits no serendipitous arbitrage, then the price of any non-constant query bundle $\bQ$ is equal to the price of asking for the whole database.
\end{theorem}

\label{proof:ii}
\begin{proof}
To incorporate serendipitous arbitrage in \QPS, the pricing function must be such that if there exists any database $D \in \mI$ such that $D \vdash \bQ_2 \dtr \bQ_1$, then we must have $p(\bQ_2) \geq p(\bQ_1)$.

Suppose that the database schema consists of the relations $R_1, \dots, R_k$. Consider the query bundle $\bQ^1$ that returns the whole database; we can always express this as a bundle of conjunctive queries, where each query $Q^1_i$ returns a relation $R_i$ of the schema (i.e. $Q^1_i(\vec{x}) = R_i(\vec{x})$). Consider also the query bundle $\bQ^0$ that checks whether the database is empty; we can express this as a bundle with a single query that is a union of conjunctive queries  (each query in the union is the boolean query $Q^0_i() = R_i(\vec{x})$).


Observe now that for the empty database $D_0$, $D_0 \vdash \bQ^0 \dtr \bQ$ for every query bundle $\bQ$. Indeed, since $\bQ^0(D_0) = False$ the data buyer knows that the database is empty and thus can determine the answer for any query bundle $\bQ$. In this case, because of the serendipitous arbitrage, we have to enforce that
$p(\bQ^0) \geq p(\bQ)$ for every query bundle $\bQ$. 

Next consider any query bundle $\bQ$ that is not constant. Then, there must exist a database $D \in \mI$ such that for some query $Q \in \bQ$, $Q(D) \neq Q(D_0)$. But in this case the data buyer knows that $D \neq D_0$, and thus can determine that $\bQ^0(D) = True$. Thus, $D \vdash \bQ \dtr \bQ^0$, and because of serendipitous arbitrage we must have $p(\bQ) \geq p(\bQ^0)$.

We have just shown that for every query bundle $\bQ$ that is not-constant, $p(\bQ) = p(\bQ^0)$. This implies that if we require that serendipitous arbitrage does not exist, every query bundle must have exactly the same price, and in particular the price of the whole database, which is equal to $p(\bQ^1)$.
%
\end{proof}

\subsection{How to Find a Pricing Function}

To characterize the structure of instance-independent pricing functions, we exploit the fact that we can equivalently view a query as a {\em partition} of the set of possible databases $\mI$. 

\subsubsection{The Partition Lattice}

Fix some query language $\mathcal{L}$. Recall that for a query bundle $\bQ \in B(\mathcal{L})$, $\mP_{\bQ}$ is the partition that is induced by the following {\em equivalence relation}: $D \sim D'$ iff $\bQ(D) = \bQ(D')$.

\begin{lemma} \label{lem:partition:mon}
Let $\bQ_1, \bQ_2 \in \mL$ be two query bundles. The following are equivalent:

\begin{packed_enum}
\item $\bQ_1 \dtr \bQ_2$
\item $\mP_{\bQ_1} \succeq \mP_{\bQ_2}$, i.e. $\mP_{\bQ_1}$ refines $\mP_{\bQ_2}$
\end{packed_enum}
\end{lemma}

\begin{proof}
$1 \implies 2$. Suppose $B_1 \in \mP_{\bQ_1} $. Let $D \in B_1$ and let $B_2$ the unique block in $\mP_{\bQ_2} $ for which $D \in B_2$. We will show that $B_1 \subseteq B_2$. Indeed, consider any other $D' \in B_1$. Then, $\bQ_1(D') = \bQ_1(D)$. Since $\bQ_1 \dtr \bQ_2$, we have $\bQ_2(D') = \bQ_2(D)$	and thus $D' \in B_2$ as well.

$2 \implies 1.$ Let $D', D'' \in \mI$ such that $\bQ_1(D') = \bQ_1(D'')$. Then, $D', D''$ both belong in the same block $B_1 \in \mP_{\bQ_1} $. Since $\mP_{\bQ_1} $ is a refinement of $\mP_{\bQ_2} $, there exists a block $B_2 \in \mP_{\bQ_2} $ such that $B_1 \subseteq B_2$. Thus, $D', D''$ belong in the same block in $\mP_{\bQ_2} $, which implies that $\bQ_2(D') = \bQ_2(D'')$.
\end{proof}

The refinement relation defines a partial order on the set $\Pi_{\mI}^{\mathcal{L}}$ of all partitions of $\mI$ induced by any bundle $\bQ \in B(\mathcal{L})$. An equivalent way to define the partial order is through the {\em distinction set} of a partition
$ dit(\mP) = \bigcup_{B,B' \in \mP: B \neq B'} B \times B'.$
Intuitively, the distinction set contains all pairs of elements that are not in the equivalence relation. It is straightforward to see that
$ dit(\mP_{\bQ}) =  \setof{(D', D'') \in \mI \times \mI}{\bQ(D') \neq \bQ(D'')}$.
Furthermore, $\mP_1 \succeq \mP_2$ if and only if $dit(\mP_1) \supseteq dit(\mP_2)$ and thus one can  use the inclusion of the distinction sets to define a partial order on the partitions.

The partial order induced by $\succeq$ on $\Pi_{\mI}^{\mathcal{L}}$ forms a {\em join-semilattice}. The bottom element of the semilattice is the partition $\{\mI\}$, which corresponds to a query that returns a constant answer. The top element is the partition where each block is a singleton set: this corresponds to a query that informs about the whole database. The {\em join} $\mP_1 \vee \mP_2$  is a new partition whose blocks are the non-empty intersections of any two blocks from $\mP_1, \mP_2$. The lemma below proves that the algebraic structure we defined is indeed a semilattice.

\begin{lemma} \label{lem:partition:sub}
Let $\bQ = \bQ_1, \bQ_2$, where $\bQ_1, \bQ_2 \in B(\mL)$. Then, $\mP_{\bQ} = \mP_{\bQ_1} \vee \mP_{\bQ_2}$.
\end{lemma}

\begin{proof}
For some database $D \in \mI$, let $B$ be the unique block that contains $D$ in $\mP_{\bQ}$, and $T$ the corresponding block in $\mP_{\bQ_1} \vee \mP_{\bQ_2}$. Note that the $\bQ \in B(\mL)$ since both $\bQ_1, \bQ_2 \in B(\mL)$. We now show that $B = T$. Indeed, let $D' \in B$. Then, $\bQ(D) = \bQ(D')$, which implies that $\bQ_i(D) = \bQ_i(D')$ for $i=1,2$. Thus, there exists a set $B_1 \in \mP_{\bQ_1}$ (respectively $B_2 \in \mP_{\bQ_2}$) that contains both $D,D'$. But then $\{D,D'\} \subseteq B_1 \cap B_2 \subseteq T$, so $D' \in T$. The reverse direction is similar.
\end{proof}

We now present an example to illustrate the mechanics of how partition lattice works.

\begin{example} \label{ex:ii}
 Consider the relation $R(\underline{A},B)$ in~\autoref{ex:intro} with $n=2$ tuples. The partition  join-semilattice is depicted in~\autoref{fig:pl}, where we encode the database $D_{ij}$ with its decimal representation (for example, the element $0$ corresponds to the database $D_{00}$).

As before, we consider the query $Q(x) = R(a_1,x)$. It is easy to see that $\mP_{Q} = \{\{D_{00}, D_{01}\}, \{D_{10}, D_{11}\}\}$, which is the element $01\vert 23$ in the lattice. For the query $Q'() = R(x,0)$,  $\mP_{Q'} = \set{ \{D_{00}, D_{01}, D_{10} \}, \{D_{11}\}}$, which is the element $012 \vert 3$. One can see in the lattice that the join of the two partitions is the element $01\vert2\vert3$. The reader can check that the bundle $(Q,Q')$ indeed induces the partition $01\vert2\vert3$.
\end{example} 

\begin{figure}
\centering
\begin{tikzpicture}[scale=0.7, transform shape]
  \node (max) at (0,-8) { $0123$ };
  \node (a) at (-6,-5) { $03\vert12$ };
  \node (b) at (-4,-5) { $0\vert123$ };
  \node (c) at (-2,-5) { $013\vert2$ };
  \node (d) at (0,-5) { $02\vert13$ };  
  \node (e) at (2,-5) { $012\vert3$ };
  \node (f) at (4,-5) { $023\vert1$ };
  \node (g) at (6,-5) { $01\vert23$ };
  
   \node (h) at (-6,-2) { $0\vert12\vert3$ };
  \node (i) at (-4,-2) { $03\vert1\vert2$ };
  \node (j) at (-2,-2) { $0\vert13\vert2$ };  
  \node (k) at (2,-2) { $02\vert1\vert3$ };
  \node (l) at (4,-2) { $01\vert2\vert3$ };
  \node (m) at (6,-2) { $0\vert1\vert23$ };
  \node (min) at (0,1) { $0\vert1\vert2\vert3$ };
  
  \draw (min) -- (h) -- (a) -- (max)
  (max) -- (b)
  (min) -- (i) -- (a)
  (min) -- (j) -- (c) -- (max)  
  (min) -- (k) -- (d) -- (max)
  (min) -- (l) -- (e) -- (max)
  (min) -- (m) -- (f) -- (max)
  (min) -- (m) -- (g) -- (max)
  (h) -- (b) 
  (h) -- (e)
  (i) -- (c)
  (i) -- (f)
  (j) -- (b)
  (j) -- (d)
  (k) -- (e)
  (k) -- (f)
  (l) -- (c)
  (l) -- (g)
  (m) -- (b) ;
\end{tikzpicture}
\caption{The partition join-semilattice for Example~\ref{ex:ii}.} \label{fig:pl}
\end{figure}
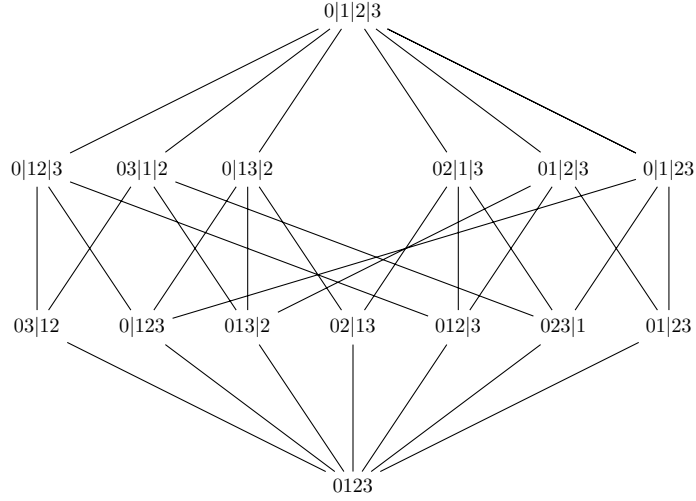

If we define the partial order as the inclusion of distinction sets, $dit(\mP_{\bQ}) = dit(\mP_{\bQ_1} \vee \mP_{\bQ_2} ) = dit(\mP_{\bQ_1}) \cup dit(\mP_{\bQ_2})$, the join operator is simply the union of the distinction sets. 

\begin{lemma} \label{lem:equiv:dtr:ii}
Let $\bQ_1, \bQ_2$ be two query bundles.
The following two statements are equivalent:
\begin{packed_enum}
\item $\bQ_2 \dtr \bQ_1$
\item $dit(\mP_{\bQ_2}) \supseteq dit(\mP_{\bQ_1})$
\end{packed_enum}
\end{lemma}

\begin{proof}
$1 \Rightarrow 2$. Consider a pair $(D',D'') \notin dit(\mP_{\bQ_2})$. Then, it must be that $\bQ_2(D') = \bQ_2(D'')$. By the definition of information-theoretic determinacy, this implies that $\bQ_1(D') = \bQ_1(D'') $, and thus it will be that $(D',D'') \notin dit(\mP_{\bQ_1})$.  \\

\noindent $2 \Rightarrow 1$. Consider databases $D',D''$ such that $\bQ_2(D') = \bQ_2(D'')$. Then, by definition $(D',D'') \notin dit(\mP_{\bQ_2})$, which implies that $(D',D'') \notin \mS_{\bQ_1}$. But then we have that $\bQ_1(D') = \bQ_1(D'')$. Thus, $\bQ_2 \dtr \bQ_1$. 
\end{proof}

\begin{lemma} \label{lem:bundle:ii}
Let $\bQ = \bQ_1, \bQ_2$ be a query bundle. Then
$ dit(\mP_{\bQ}) =  dit(\mP_{\bQ_1}) \cup dit(\mP_{\bQ_2})$.
\end{lemma}

\begin{lemma}
Let $ p(\bQ) = f(dit(\mP_{\bQ}))$. If $f$ is monotone, then $p$ has no information arbitrage.
\end{lemma}

\begin{proof}
Consider two query bundles $\bQ_1, \bQ_2$ such that $\bQ_2 \dtr \bQ_1$. From Lemma~\ref{lem:equiv:dtr:ii} this is equivalent to $dit(\mP_{\bQ_1}) \subseteq dit(\mP_{\bQ_2})$. Since $f$ is monotone, we have
$ p(\bQ_2) = f(dit(\mP_{\bQ_2})) \geq f(dit(\mP_{\bQ_1})) =  p(\bQ_1).$
\end{proof}

\begin{lemma} 
Let $ p(\bQ) = f(dit(\mP_{\bQ}))$. If $f$ is subadditive, then $p$ has no bundle arbitrage.
\end{lemma}

\begin{proof}
Let $\bQ = \bQ_1, \bQ_2$. From Lemma~\ref{lem:bundle:ii},  $ dit(\mP_{\bQ}) =  dit(\mP_{\bQ_1}) \cup dit(\mP_{\bQ_2})$.
Since $f$ is subadditive, we  have:
\begin{align*}
p(\bQ)  = f(dit(\mP_{\bQ}))  & \leq f(dit(\mP_{\bQ_1})) + f(dit(\mP_{\bQ_2})) 
 =  p(\bQ_1)  + p(\bQ_2) 
\end{align*}
This concludes the proof.
\end{proof}

\begin{lemma}
Let $p$ be an answer-independent pricing function with no information arbitrage. 
Then, $p$ must be of the form $p(\bQ) = f(dit(\mP_{\bQ}))$, where $f$ is a set function over $\mathcal{I} \times \mathcal{I}$.
\end{lemma}

\begin{proof}
To prove the lemma, we will show that for any two queries $\bQ_1, \bQ_2$,  $dit(\mP_{\bQ_1})= dit(\mP_{\bQ_2})$ implies that they have the same price. From Lemma~\ref{lem:equiv:dtr:ii} we obtain that $\bQ_2 \dtr \bQ_1$. Since $p$ has no information arbitrage, it must be that $p(\bQ_2) \geq p(\bQ_1)$. Using a symmetric argument, we can also prove that $p(\bQ_2) \leq p(\bQ_1)$, which implies that the prices are indeed the same: $p(\bQ_2) = p(\bQ_1)$.
This proves the existence of a function $f$.
\end{proof}

\subsubsection{A Characterization of Arbitrage-Free \QPS}

We now consider the family of instance-independent pricing functions of the form $ p(\bQ) = f(\mP_{\bQ})$, where $f: \Pi_{\mI}^\mL \rightarrow  \mathbb{R}_+$ is a function that maps a partition to the positive real numbers. 

\begin{theorem} \label{thm:arb:1}
Let $p$ be an instance-independent pricing function. Then, the two statements are equivalent:
\begin{packed_enum}
\item $p$ has no information arbitrage.
\item $p(\bQ) = f(\mP_{\bQ})$, where $f$ is a monotone function over $\Pi_\mI^\mL$.
\end{packed_enum}
\end{theorem}

\begin{proof}
$2 \implies 1$. Say $\bQ_1 \dtr \bQ_2$. Then, by~\autoref{lem:partition:mon}, $\mP_{\bQ_1} \succeq \mP_{\bQ_2} $. By monotonicity of $f$, we have $ p(\bQ_1) = f(\mP_{\bQ_1}) \geq f(\mP_{\bQ_2}) =  p(\bQ_2)$.

$1 \implies 2$. 
We will first show that if $\mP_{\bQ_1} = \mP_{\bQ_2}$, then $\bQ_1, \bQ_2$ have the same price. Indeed, since $\mP_{\bQ_1} \succeq \mP_{\bQ_2}$, by Lemma~\ref{lem:partition:mon} we have $\bQ_1 \dtr \bQ_2$, which implies $p(\bQ_1) \geq p(\bQ_2)$ since $p$ is information arbitrage-free. Similarly, $p(\bQ_1) \leq p(\bQ_2)$. Thus, there must exist some function $f$ such that $p(\bQ) = f(\mP_{\bQ})$.

Next, assume that $\mP_1 \succeq \mP_2$ for $\mP_1, \mP_2 \in \Pi_\mI^\mL$. Then, we can find $\bQ_1, \bQ_2$ such that $\mP_{\bQ_i} = \mP_i$ for $i=1,2$. From Lemma~\ref{lem:partition:mon} then, $\bQ_1 \dtr \bQ_2$. Thus, $f(\mP_1) = p(\bQ_1) \geq p(\bQ_2) = f(\mP_2)$.
\end{proof}

\begin{theorem} \label{thm:arb:2}
Let $p(\bQ) = f(\mP_\bQ)$ be an instance-independent pricing function, where $f$ is a function over $\Pi_\mI^\mL$. Then, the two statements are equivalent:
\begin{packed_enum}
\item $p$ has no bundle arbitrage.
\item $f$ is subadditive over $\Pi_\mI^\mL$.
\end{packed_enum}
\end{theorem}

\begin{proof}
$2 \implies 1$. Let $\bQ = \bQ_1, \bQ_2$. By Lemma~\ref{lem:partition:sub}, $\mP_{\bQ} =  \mP_{\bQ_1} \vee \mP_{\bQ_2} $. Since $f$ is subadditive over the join-semilattice:
\begin{align*} 
p(\bQ) = f(\mP_{\bQ}) = f(\mP_{\bQ_1} \vee \mP_{\bQ_2} ) 
 \leq f(\mP_{\bQ_1}) + f(\mP_{\bQ_2})=  p(\bQ_1) + p(\bQ_2) 
\end{align*}

$1 \implies 2$. 
Let $\mP = \mP_1 \vee \mP_2$. Then, we can find $\bQ_1, \bQ_2$ such that $\mP_{\bQ_i} = \mP_i$ for $i=1,2$. From Lemma~\ref{lem:partition:sub}, if $\bQ = \bQ_1, \bQ_2$, then $\mP_{\bQ} = \mP$. Thus, 
\begin{align*} 
f(\mP) = p(\bQ) \leq  p(\bQ_1) + p(\bQ_2) = f(\mP_1) + f(\mP_2).
\end{align*}
This concludes the proof. 
\end{proof}


\begin{corollary}
Let $f$ be a monotone and subadditive function over $\Pi_\mI^\mL$. Then, $p(\bQ) = f(\mP_\bQ)$ is an instance-independent pricing function that has no bundle or information arbitrage.
\end{corollary}

Alternatively, we could also define the pricing function as $p(\bQ) = f(dit(\mP_{\bQ}))$. Using the same type of arguments, we can show:

\begin{corollary}
Let $f$ be a monotone and subadditive set function. Then, $p(\bQ) = f(dit(\mP_\bQ))$ is an instance-independent pricing function that has no bundle or information arbitrage.
\end{corollary}



\subsection{Construction of Pricing Functions From Answer-Dependent Prices}

We show first how we can design an instance-independent pricing function $p(\bQ)$ starting from an answer-dependent function $p(\bQ, E)$. Given a query bundle $\bQ$, the idea is to construct a vector of all prices $p(\bQ, \bQ(D))$ for all databases $D \in \mathcal{I}$. Formally, we define the {\em price vector} 
$ \vec{p}(\bQ) = \langle p(\bQ,\bQ(D)) \mid D \in \mathcal{I} \rangle$.
Then we can obtain an instance-independent pricing function by computing another function $g: \mathbb{R}_+^{|\mathcal{I}|} \rightarrow \mathbb{R}_+$ over the above vector, such that $p(\bQ) = g(\vec{p}(\bQ))$. The next lemma describes the conditions for $g$ under which the arbitrage-free property carries over.

\begin{lemma} \label{lem:average}
Let $p(\bQ, E)$ be an arbitrage-free pricing function. If $g$ is a monotone and subadditive function, then $p(\bQ) = g(\vec{p}(\bQ))$ is an arbitrage-free instance-independent function. 
\end{lemma}

\begin{proof}
We first prove the information arbitrage property. Suppose that $\bQ_1 \dtr \bQ_2$. Then, for every database $D$ we have $D \vdash \bQ_1 \dtr \bQ_2$, which implies $p(\bQ_1, \bQ_1(D)) \geq p(\bQ_2, \bQ_2(D))$. Since $g$ is monotone and the price vector for $\bQ_2$ is smaller everywhere than the vector for $\bQ_1$, we have that $p(\bQ_1) \geq p(\bQ_2)$.

We next prove the bundle arbitrage property. Let $\bQ = \bQ_1, \bQ_2$. Since $p(\bQ, E)$ is bundle-arbitrage-free for every $E = \bQ(D)$, we have that for every database $D$, $p(\bQ, \bQ(D)) \leq p(\bQ_1, \bQ_1(D))  + p(\bQ_2, \bQ_2(D)) $.  Thus:
\begin{align*}
p(\bQ)  = g(\vec{p}(\bQ)) \leq g(\vec{p}(\bQ_1) + \vec{p}(\bQ_2)) 
 \leq g(\vec{p}(\bQ_1)) + g(\vec{p}(\bQ_2))
\end{align*} 
where the first inequality comes from the monotonicity of $g$, and the second inequality results from the subadditivity of $g$.
\end{proof}

We next present an application of Lemma~\ref{lem:average} to obtain arbitrage-free pricing functions. 

\begin{lemma}\label{lem:ii:weight}
Let $f$ be a monotone and subadditive set function. Let $w_D$ be a non-negative weight $w_D$ to each $D \in \mI$, and denote $w_B = \sum_{D \in B} w_D$. Then, the pricing functions 
$p_1(\bQ)  =  \max_{B \in \mP_\bQ} \{f(\mI \setminus B)\}$ and $p_2(\bQ) = \sum_{B \in \mP_\bQ} w_B \cdot  f(\mI \setminus B)$ are arbitrage-free.
%
\end{lemma}

\begin{proof}
Since $f$ is a subadditive and monotone set function, $p(\bQ,E) = f(\dagr{\bQ}{E})$ is arbitrage-free. 
For the function $p_1$, we apply Lemma~\ref{lem:average} with $g$ being the max norm.  In this case, we have
\begin{align*} 
p(\bQ)  = \max_{D \in \mI} \{p(\bQ,\bQ(D))\} = \max_{E} \{p(\bQ,E)\} 
 = \max_{E} \{ f(\bar \mS_\bQ(E))\} 
 = \max_{B \in \mP_\bQ} \{f(\mI \setminus B)\}
\end{align*}
For $p_2$, we apply Lemma~\ref{lem:average} with $g$ being the weighted norm $g(\vec{x}) = \sum_D w_D x_D$. Then:
\begin{align*} 
p(\bQ) & = \sum_{D \in \mI} w_D \cdot p(\bQ,\bQ(D)) 
 = \sum_{E} \left(\sum_{D:\bQ(D)=E} w_D \right) \cdot f(\dagr{\bQ}{E}) 
 = \sum_{B \in \mP_\bQ} w_B \cdot  f(\mI \setminus B)
\end{align*} 
This concludes the proof.
\end{proof}

\begin{example}
Consider the function $p_2$ with equal weights $w_D = 1$ and the set function $f(A) = |A|$. The resulting  arbitrage-free function is 
$p(\bQ) = \sum_{B \in \mP_\bQ} |B|(|\mI| - |B|)= |dit(\mP_\bQ)|$, 
which sets the price to be the size of the distinction set.
\end{example}

If $\sum_{D} w_D = 1$ for $p_2$, one can interpret the weights as a probability distribution over the set of databases $\mI$. In this case, we can write $p_2(\bQ) = \mathbb{E}_{B \in \mP_\bQ} [f(\mI \setminus B)]$, 
where each block $B$ has probability $w_B$. In other words, the pricing function is the expected price over all answer-dependent prices.  
The converse of Lemma~\ref{lem:average} does not hold: it is possible for $p(\bQ)$ to be arbitrage-free, and for some database $D$ it may not be the case. As we will see next, this allows us to construct arbitrage-free functions that are based on measures of uncertainty.

\subsection{Construction of Pricing Functions From Uncertainty Measures} 
\label{sec:entropy}

In this section, we describe arbitrage-free pricing functions that do not originate from answer-dependent functions. To construct such functions, we switch to a probabilistic view of the problem and then apply information-theoretic tools that are used to measure uncertainty. For the remainder of this section, we assume that each database $D$ is associated with a probability $p_D$. We denote by $X$ the random variable such that $P(X = D) = p_D$ and let $p_E = \sum_{D: \bQ(D)=E} p_D$. The detailed proofs in this section are presented in~\autoref{proof:ii}. \\

\noindent {\bf Shannon Entropy.}
The first measure of uncertainty we apply is the most commonly used form of entropy, and was proposed in~\cite{LK14} as a pricing function. In the answer-dependent context, we defined the price as the information gain after the output $E$ has been revealed. Since in this setting the price is independent of the output, we define the price as the {\em expected information gain} over all possible outcomes. Formally:
\begin{align} \label{eq:shannon:entropy}
p^H(\bQ) & = H(X) - \sum_E p_E \cdot H(X \mid \bQ(X) = E)
\end{align}
Equivalently, we can also express the price as
\begin{align*} 
p^H(\bQ)  = H(X) - H(X \mid \bQ(X))  = I(X ; \bQ(X)) 
 = H(\bQ(X)) - H(\bQ(X) \mid X) = H(\bQ(X))
\end{align*}
where $I(X;Y)$ is the {\em mutual information} between the random variables $X$ and $Y$. \cite{LK14} proves that $p^H$ is both bundle and information arbitrage-free, using the subadditivity of entropy and the data-processing inequality respectively. It is instructing to write $p^H$ as
$
p^H(\bQ)  
= - \sum_{S \in \mP_\bQ} p_S \cdot \log p_S 
 = \sum_D p_D \cdot p(\bQ, \bQ(D))
$
where $p(\bQ, E) = - \log \left( p_E \right)$ is now an answer-dependent pricing function. Notice that $p(\bQ,E)$ has no information arbitrage, and thus by applying Lemma~\ref{lem:average} we get an alternative proof that $p^H$ is information arbitrage-free. However, $p(\bQ,E)$ can have bundle arbitrage, and thus we cannot apply Lemma~\ref{lem:average} to show the subadditivity property as well: entropy is subadditive only in expectation. This example demonstrates that the converse of Lemma~\ref{lem:average} does not hold. \\

\noindent {\bf Tsallis Entropy.}
For a real number $q > 1$, the 
{\em Tsallis entropy}~\cite{tsallis-entropy}, or {\em $q$-entropy}, of a random variable $X$ is defined as
$ S_q(X) = \frac{1}{q-1} \cdot \left( 1- \sum_x P(X=x)^{q-1} \right)$.
Tsallis entropy is a generalization of Shannon entropy, since $\lim_{q \rightarrow 1} S_q(X) = H(X)$. We define the price as the Tsallis entropy of $\bQ(X)$:
\begin{align} \label{eq:tsallis:entropy}
p^T(\bQ)  = S_q(\bQ(X))  =\sum_{S \in \mP_\bQ} \frac{p_S}{q-1} \cdot (1-p_S^{q-1})
\end{align}

\begin{lemma}\label{lem:ii:tsallis}
The pricing function $p^T$ defined in Equation~\eqref{eq:tsallis:entropy} is arbitrage-free for $q>1$. \end{lemma}

\begin{proof}
To show that $p^T$ has no information arbitrage, notice that we can write the pricing function as $\frac{1}{q-1} \sum_D p_D f(D_\bQ)$, where $f$ is the set function $f(S) = 1-p_{S}^{q-1}$. Observe that $f$ is a decreasing function. Suppose now we have two partitions such that $\mP_1 \succeq \mP_2$. For some database $D$, let $S_1 \in \mP_1$ the set that contains $D$, and similarly define $S_2 \in \mP_2$. Since $S_1 \subseteq S_2$, we have $f(S_1) \geq f(S_2)$. Summing over all databases proves that the function is indeed information arbitrage-free.

To show that $p^T$ has no bundle arbitrage, we will use the property that $q$-entropy is subadditive for any $q > 1$~\cite{tsallis:subadditive}. We then can write:
\begin{align*} 
p^T(\bQ_1, \bQ_2) & = S_q(\bQ_1(X), \bQ_2(X)) 
 \leq S_q(\bQ_1(X)) +S_q(\bQ_1(X)) 
= p^T(\bQ_1) + p^T(\bQ_2)
\end{align*}
This concludes the proof that $p^T$ has no arbitrage.
\end{proof}

\noindent {\bf Guessing Entropy.} 
The guessing entropy measures the average number of successive guesses required by an optimum strategy until we correctly guess the value of the random variable $X$ (in our case the underlying database $D$). The guessing entropy was first introduced in~\cite{guessing-entropy}, and subsequently used in~\cite{kopf2007information} in the context of measuring leakage in side-channel attacks.
To compute the guessing entropy of $X$, suppose that we have ordered the databases in decreasing order of their probabilities, i.e. such that $p(X=D_i) \geq p(X=D_j) $ whenever $i \leq j$. Then, we define the {\em guessing entropy} as $G(X) = \sum_i i \cdot p_{D_i}$. The price is now defined as the initial entropy minus the expected conditional guessing entropy $G(X \mid \bQ(X) =E)$:
\begin{equation}  \label{eq:guessing:entropy}
p^G(\bQ) = G(X) - \sum_{E} p_E \cdot G(X \mid \bQ(X) = E)
\end{equation}

\begin{lemma}\label{lem:ii:guessing}
The pricing function $p^G$ defined in Equation~\eqref{eq:guessing:entropy} is arbitrage-free. 
\end{lemma}

To prove that the guessing entropy is arbitrage-free, it will be convenient to rewrite the above pricing function in a simpler form. For a given set $S \subseteq \mI$, denote by $i_{S}(D)$ the position of $D$ in an ordering of the elements in $S$ in decreasing probability. Then, we can write:
\begin{align*}
p^G(\bQ)  = \sum_{D} p_D \cdot i_{\mI}(D) - \sum_{S \in \mP_\bQ} \sum_{D \in S} p_D \cdot i_{S}(D)  = \sum_{D} p_D \cdot ( i_{\mI}(D) -  i_{[D]_\bQ}(D)).
\end{align*}

We can now use the above form to prove that the guessing entropy is a well-behaved pricing function. \\

\begin{proof}
We will prove the lemma by using the characterization of arbitrage in terms if the monotonicity and subadditivity of the function applied on elements of the partition lattice. 

For monotonicity, suppose that $\mP_1 \succeq \mP_2$. Consider a database $D$ that belongs in set $S_1 \in \mP_1$ and $S_2 \in \mP_2$. Since $\mP_1$ refines $\mP_2$, it must be that $S_1 \subseteq S_2$. But then, the index of $D$ in $S_2$ will be at least as large (since the set has a superset of elements). Thus, $i_{S_1}(D) \leq i_{S_2}(D)$, which implies in turn that $i_{\mI}(D) - i_{S_1}(D) \geq i_{\mI}(D) - i_{S_2}(D)$. Summing over all databases $D \in \mI$ obtains the desired inequality.

To prove the subadditivity property, let $\mP = \mP_1 \vee \mP_2$. Consider a database $D$ that belongs in $S \in \mP$, and also $S_1 \in \mP_1, S_2 \in \mP_2$. Let us now denote by $C$ the set of databases that have index $\geq i_{\mI}(D)$ in the set $\mI$. The key observation is that the index of $D$ in any set will depend only on $C$. By the construction of $C$, we then have that for every set $S$, $i_{S}(D) = |S \cap C|$. Since $S = S_1 \cap S_2$, we have $C \cap S = (C \cap S_1) \cap (C \cap S_2)$, or equivalently $C \setminus S = (C \setminus S_1) \cup (C \setminus S_2)$. Now:
\begin{align*}
i_{\mI}(D) - i_{S}(D) & = |C \cap \mI| - |C \cap S| =  |C \setminus S| \\
& \leq |C \setminus S_1| + |C \setminus S_2| \\
& = (|C| - |C \cap S_1|) + (|C| - |C \cap S_2|)  \\
& = (i_{\mI}(D) - i_{S_1}(D)) + (i_{\mI}(D) - i_{S_2}(D))
\end{align*}
Summing over all databases $D \in \mI$ proves the desired inequality for subadditivity.
\end{proof}

\noindent {\bf Min-Entropy.} 
We apply here the notion of min-entropy, as it was introduced in~\cite{info-flow} to quantify information flow. The {\em min-entropy} of a random variable is $H_\infty(X) = - \log (\max_x P(X=x))$. The conditional min-entropy is defined as $H_\infty(X \mid Y) = - \log (\sum_y P(Y=y) \cdot  \max_x P(X = x \mid Y=y))$. Then, we can construct the price of a query as follows:
\begin{align} \label{eq:min:entropy}
p^M(\bQ)  = H_\infty(X) - H_\infty(X \mid \bQ(X))   
 =  - \log (\max_D p_D) + \log \left( \sum_E \max_{D: \bQ(D)=E} p_D \right) 
\end{align}

\begin{lemma}\label{lem:ii:min:entropy}
The pricing function $p^M$ defined in Equation~\eqref{eq:min:entropy} has no information arbitrage. 
\end{lemma}

\begin{proof}
We will show that the function is monotone on the partition lattice. For partitions $\mP_1, \mP_2$ such that $\mP_1 \succeq \mP_2$, it suffices to show that $\sum_{S \in \mP_1} \max_{D \in S} \{ p_D \} \geq \sum_{S \in \mP_2} \max_{D \in S} \{p_D\}$. But now notice that for each set $S_2 \in \mP_2$, there exists a unique set $S_1 \in \mP_1$, such that $S_1 \subseteq S_2$ and $\max_{D \in S_1}  \{ p_D \}  = \max_{D \in S_2}  \{ p_D \} $.
\end{proof}

The min-entropy is not in general bundle arbitrage-free, as we show in Example~\ref{example:ba:violation} below.

\begin{example}\label{example:ba:violation}
Let $R(A,B)$ be a binary relation and assume that $\mI$ consists of the following four databases: $D_{00} = \{(a,0), (b,0) \}$, $D_{01} = \{(a,0), (b,1) \}$, $D_{10} = \{(a,1), (b,0) \}$, $D_{11} = \{(a,1), (b,1) \}$. We set the probability to $0.7$ for $D_{00}$, and $0.1$ for the other three databases.

The min-entropy of the initial distribution is $H_\infty(X) = -\log (0.7)$. Now, let $Q_1 = \sigma_{A=a}(R)$ and $Q_2 = \sigma_{A=b}(R)$. One can see that  $\mP_{Q_1} = \{ \{D_{00}, D_{01}\}, \{ D_{10}, D_{11}\} \}$ and $\mP_{Q_2} = \{ \{D_{00}, D_{10}\}, \{ D_{01}, D_{11}\} \}$. Thus, $H_\infty(X \mid Q_1(X)) = - \log (0.7+0.1) = -\log(0.8)$, from which we obtain $p^M(Q_1) = \log(8/7)$. Similarly, $p^M(Q_2) = \log(8/7)$. For the bundle $\bQ = Q_1, Q_2$, observe that the partition is $\mP_\bQ = \{ \{D_{00} \}, \{ D_{01}\}, \{ D_{10} \}, \{D_{11}\} \}$. Thus, $p^M(\bQ) = -\log(0.7) + \log(1) = \log(10/7)$. Notice finally that $\log(10/7) > 2\log(8/7)$, hence violates the bundle arbitrage condition. 
\end{example}

However, it becomes so when the initial distribution is uniform. Let $n= |\mI|$, in which case $p_D = 1/n$ for each database. Then, it is straightforward to see that the resulting function is the logarithm of the number of sets in the partition $\mP_\bQ$.
\begin{align} \label{eq:min:entropy:uniform}
p^{MU}(\bQ) = \log(n) + \log(|\mP_\bQ|/n) = \log(|\mP_\bQ|)
\end{align}
%

\begin{lemma}\label{lem:min:entropy:uniform}
The pricing function $p^{MU}(\bQ)$ defined in Equation~\eqref{eq:min:entropy:uniform} is arbitrage-free.
\end{lemma}

\begin{proof}
Since we have already proved that $p^{MU}$ has no information arbitrage, we now show that it has no bundle arbitrage as well. Let $\mP = \mP_1 \vee \mP_2$. Since each set in $\mP$ is the unique intersection of one set from $\mP_1$ and one set from $\mP_2$, we have that $|\mP| \leq |\mP_1| \cdot |\mP_2|$. The desired result is obtained taking the $\log$ on each side of the equation.
\end{proof}

\noindent {\bf $\beta$-Success Rate.} This information measure, first introduced in~\cite{success-rate}, captures the expected success of guessing the database $D$ with $\beta$ tries. We will consider here only the case where the probability distribution is uniform, in which case the pricing functions becomes:
\begin{align} \label{eq:success:rate}
p^{\beta}(\bQ) = \log \left( \sum_{S \in \mP_\bQ} \min \{ \beta, |S| \} \right)
\end{align} 
Observe that for $\beta=1$ we have $p^\beta(\bQ) = p^{MU}(\bQ)$, hence this generalizes uniform min-entropy.

\begin{lemma}\label{lem:ii:beta}
The pricing function $p^{\beta}(\bQ)$ defined in Equation~\eqref{eq:success:rate} is arbitrage-free.
\end{lemma}

\begin{proof}
We will use again the characterization of arbitrage in terms of the monotonicity and subadditivity of the function on elements of the partition lattice. 

Consider partitions $\mP_1 \succeq \mP_2$. We need to show that $\sum_{S \in \mP_1} \min \{\beta, |S|\} \geq \sum_{S' \in \mP_2} \min \{\beta, |S'|\}$. For a set $S' \in \mP_2$, consider all the sets $S \in \mP_1$ such that $S \subseteq S'$. We will show that $\sum_{S \in \mP_1: S \subseteq S'} \min \{\beta, |S| \} \geq \min \{\beta, |S'| \}$. Indeed, if some set $S \subseteq S'$ we have $\beta = \min \{ \beta, |S| \}$, then it must also be $\beta = \min \{ \beta, |S'| \}$, hence the inequality holds. Otherwise, the left hand side becomes equal to $|S'|$, in which case the inequality holds since trivially $|S'| \geq \min \{\beta, |S'| \}$.

To show no bundle arbitrage, let $\mP = \mP_1 \vee \mP_2$. We can now write:
\begin{align*}
& (\sum_{S \in \mP_1} \min \{\beta, |S|\}) \cdot (\sum_{S' \in \mP_2} \min \{\beta, |S'|\}) \geq \\
 \geq & \sum_{S_1 \cap S_2 \neq \emptyset} \min \{\beta, |S_1|\} \cdot \min \{\beta, |S_2|\}  \\
\geq & \sum_{S_1 \cap S_2 \neq \emptyset} \min \{\beta, |S_1|\} 
\geq \sum_{S_1 \cap S_2 \neq \emptyset} \min \{\beta, |S_1 \cap S_2|\}  \\
  = & \sum_{S \in \mP} \min \{\beta, |S|\} 
\end{align*}
The desired inequality is obtained by taking the logarithm of both sides.
\end{proof}

We should finally mention that several other entropy measures have been discussed in the broader literature. The Renyi entropy~\cite{renyi} is a generalization of both the Shannon entropy and the min-entropy. However, it is not subadditive, and thus not applicable as an arbitrage-free pricing function. Worst-case entropy measures~\cite{kopf2007information} can also be applied to measure information leakage, but they are also prone to bundle arbitrage.

\begin{table}
\centering
\begin{tabular}{|l|l|}\hline
Shannon Entropy & $p^H(\bQ) = \frac{1}{n} \sum_{B \in \mP_\bQ} |B| \log|B|$ \\[2ex]\hline
Guessing Entropy & $p^G(\bQ) = \frac{1}{2n}\left(n^2 - \sum_{B \in \mP_\bQ}  |B|^2\right)$  \\[2ex]\hline
Min-Entropy & $p^{MU}(\bQ) = \log \vert \mP_Q \vert$  \\[2ex]\hline
Tsallis Entropy & $p^T(\bQ)  = \frac{1}{q-1} \left(1 - \sum_{B \in \mP_\bQ}  (\frac{|B|}{n})^{q-1}\right)$  \\[2ex]\hline
$\beta$-Success Rate & $p^\beta(\bQ)  = \log \left( \sum_{B \in \mP_\bQ} \min \{ \beta, |B| \} \right) $  \\[2ex]\hline
\end{tabular}
\caption{The price of a query bundle $\bQ$ according to various entropy measures for the case of uniform probability distributions. We denote $n = |\mI|$.}
\label{tbl:entropy}
\end{table}

\section{Computing the Pricing Function}
\label{sec:algorithm}

So far we have studied how to construct pricing functions for both \APS\ and \QPS. In this section, we focus on the complexity of computing a pricing function.

\subsection{Support Sets}

We first start by discussing an generic approach that can construct efficiently computable arbitrage-free pricing functions for any query language $\mathcal{L}$ that can be computed efficiently. The key idea behind our construction is to define the pricing function on a smaller set 
$\mC \subseteq \mI$ of our choice, which we call {\em support}. The next two lemmas  show that this restriction still provides arbitrage-free answer-dependent and instance-independent pricing functions.

\begin{lemma} \label{lem:support:1}
Let $\mC \subseteq \mI$. If $f$ is a monotone and subadditive set function, the pricing function $p(\bQ,E) = f(\dagr{\bQ}{E} \cap \mC)$ is arbitrage-free.
\end{lemma}

\begin{proof}
It suffices to show that $g(A) = f(A \cap \mC)$ is monotone and subadditive. Indeed, if $A \subseteq B$, we have $A \cap \mC \subseteq B \cap \mC$, and hence by the monotonicity of $f$ we get $g(A) = f(A \cap \mC) \leq f(B \cap \mC) = g(B)$. For subadditivity, assume $A = A_1 \cup A_2$. Then, $A \cap \mC = (A_1 \cap \mC) \cup (A_2 \cap \mC)$, and thus $g(A) = f(A \cap \mC) = f((A_1 \cap \mC) \cup (A_2 \cap \mC)) \leq f(A_1 \cap \mC) + f(A_2 \cap \mC) = g(A_1) + g(A_2)$.
\end{proof}

Given a partition $\mP$ of the set $\mI$, we define the {\em restriction of $\mP$ to $\mC$}, denoted $\mP \cap \mC$, as the set $\setof{B \cap \mC}{B \in \mP, B \cap \mC \neq \emptyset}$.

\begin{lemma} \label{lem:support:2}
Let $\mC \subseteq \mI$. If $f$ is a monotone and subadditive function on the partition semilattice,  the pricing function $p(\bQ) = f(\mP_\bQ \cap \mC)$ is arbitrage-free.
\end{lemma}

\begin{proof}
It suffices to show that $g(\mP) = f(\mP \cap \mC)$ is monotone and subadditive. Indeed, let $\mP_1 \succeq \mP_2$. Let $B_1 \in \mP_1 \cap \mC$. By the construction of the restriction, there must exist some $B_1' \in \mP_1$ such that $B_1 = B_1' \cap \mC$. Also, there exists a unique $B_2' \in \mP_2$ such that $B_1' \subseteq B_2'$. Notice that $B_2= B_2' \cap \mC \neq \emptyset$, which implies $B_2 \in \mP_2 \cap \mC$. But then, $B_1 \subseteq B_2$, so $\mP_1 \cap \mC \succeq \mP_2 \cap \mC $. The monotonicity of $g$ then follows from the monotonicity of $f$. 

For subadditivity, let $\mP = \mP_1 \vee \mP_2$. We will then show that $\mP \cap \mC = (\mP_1 \cap \mC) \vee (\mP_2 \cap \mC)$. Indeed, for any $B \in \mP$ such that $B \cap \mC \neq \emptyset$, we have that $B = B_1 \cap B_2$, where $B_1 \in \mP_1$ and $B_2 \in \mP_2$. But then $(B \cap \mC) = (B_1 \cap \mC) \cap (B_2 \cap \mC)$.
\end{proof}

The above results provide us with a method to design an efficient arbitrage-free pricing function for a query language $\mathcal{L}$. We start by choosing a support $\mC \subseteq \mI$. To compute the pricing function, we first compute $\dagr{\bQ}{E} \cap \mC$ for answer-dependent (or $\mP_\bQ \cap \mC$ for instance-independent). The observation is that we can achieve this by evaluating the query bundle $\bQ$ only on the databases $D \in \mC$. Hence, the running time of computing the price does not depend on $|\mI|$, but on $|\mC|$ and the complexity of evaluating the query bundle $\bQ$.

\begin{example}
Consider any set $\mC \subseteq \mI$. Then $p(\bQ,E) = \log |\{D \in \mC \mid \bQ(D) \neq E\}|$ is an arbitrage-free pricing function. Similarly,
$p(\bQ) = \log |\mP_\bQ \cap \mC|$ is also arbitrage-free.
\end{example}

The advantage of using support sets to construct pricing functions is that they provide us with a generic method that is independent of the language $\mathcal{L}$. On the other hand, the size and choice of the support $\mC$ is a challenging problem. We can always choose $\mC$ to contain a single database. The evaluation of the price will be very efficient, but any query will be assigned only one of two prices, and thus the pricing function will not be very successful in measuring the value of the data. If we instead choose a very large support, this leads to expensive and impractical price computation. We leave as an open research question how to choose a good support $\mC$ that is suitable for a practical implementation. \\

\subsection{The Complexity of Entropy-Based Pricing}

In a practical setting, the set $\mI$ will be given implicitly. For example, $\mI$ can be the infinite set of all databases, or the set of all subsets
of a given database $D_0$, $\mI = \{D \mid D \subseteq D_0\}$. One might think that since the
problem of determinacy (either query or data-dependent) is hard even for the class of
conjunctive queries, computing an arbitrage-free pricing function is always hard. However, as
we showed in the previous section about support sets, it is always possible to construct non-trivial pricing schemes that circumvent the computation of determinacy and thus can be computed efficiently. Here we will focus on the computational complexity for the pricing functions we introduced that are based on entropy.

The task necessary to compute an answer-dependent pricing function such as $p(\bQ,E) = \log(|\dagr{\bQ}{E}|)$, or any of the instance-independent functions in~\autoref{tbl:entropy} is the following: {\em  given a view extension $E$ and $\bQ$, compute $|\agr{\bQ}{E}|$, which is the number of databases in $\mI$ such that $\bQ(D)=E$}. If $\mI$ can be succinctly expressed as $\mI = \{D \mid D \subseteq D_0\}$, the task relates to the area of probabilistic databases. Indeed, we can view $D_0$ as a tuple-independent probabilistic database where each tuple has the same probability $1/2$. Then, we can write $|\agr{\bQ}{E}| = P(\bQ(D_0) = E) \cdot |\mI|$. 
%
%
Unfortunately, computing the probability $P(\bQ(D_0) = E)$ is in general a ${}^\#P$-hard problem (w.r.t. the size of $D_0$), even for the class of conjunctive queries~\cite{prob}. However, the task is known to be in polynomial time for certain  types of queries. For instance, in Example~\ref{ex:intro}, where $Q$ is a selection query over a single table, the size of the conflict set can be computed exactly in polynomial time.
We should note here that the problem of checking whether $\agr{\bQ}{E}$ is empty or not is equivalent to the problem of {\em view consistency}, which is shown to be NP-hard for the class of conjunctive queries~\cite{AD98} when $\mI$ ranges over all databases.




Even if $|\agr{\bQ}{E}|$ can be computed exactly, the number of blocks in the partition $\mP_\bQ$ may still be exponentially large, which would make computing the Shannon or Guessing entropy intractable. In this case, we can write the information gain as $ p(\bQ) = - \sum_{D \in \mI} \log | [D]_\bQ| $, and construct an estimator of the price that samples independently $m$ databases from $\mI$ and outputs their average: $\tilde{p}(\bQ) = \frac{1}{m} \sum_{i=1}^m \log | [D_i]_\bQ| $. In~\cite{Kopf:2010aa, entropy:approx}, the authors show that such an estimator can achieve an additive $\delta$-approximation of the price with a number of samples that is polynomial in $1/\delta, \log(|\mI|)$. We say that a pricing function is {\em $\varepsilon$-approximately arbitrage-free} if the arbitrage conditions are violated within an additive $\varepsilon$. It is straightforward to see that $\tilde{p}$ results in a $(3\delta)$-approximately arbitrage-free pricing scheme. This implies that we can compute in polynomial time an approximation of the entropy function that is as close to arbitrage-free as we would like to.

\section{Related Work}
\label{sec:related}

The problem of data pricing has been studied from a wide range of perspectives, including online markets and privacy~\cite{DworkMNS06,MS09}. \cite{jain:02} examined a variety of issues involved in pricing of information products and presented an economic approach to design of optimal pricing mechanism for online services. \cite{balazinska:11b} introduced the challenge of developing pricing functions in the context of cloud-based environments, where users can pay for queries without buying the entire dataset.  This work also outlines various research challenges, such as enabling fine-grained pricing and developing efficient and fair pricing models for cloud-based markets.

The first formal framework for query-based data pricing was introduced by Koutris et al.~\cite{KUBHS12}. The authors define the notion of arbitrage, and provide a framework that takes a set of fixed prices for views over the data identified by seller, and extends these price points to a pricing function over any query. The authors also show that evaluation of the prices can be done efficiently in polynomial time for specific classes of conjunctive queries and a restricted set of views that include only selections. Subsequently, the authors demonstrated how the framework can be implemented into a prototype pricing system called QueryMarket~\cite{KUBHS12b, KUBHS13}. Further work~\cite{LM12} discusses the pricing and complexity of pricing for the class of aggregate queries.
The work by Lin and Kifer~\cite{LK14} proposes several possible forms of arbitrage violations and integrates them into a single framework. The authors allow the queries to be randomized, and propose two potential pricing functions that are arbitrage-free across all forms. 

Data pricing is tightly connected to differential privacy~\cite{DBLP:journals/cacm/Dwork11}. Ghosh and Roth~\cite{ghosh2013selling} study the buying and selling of data by considering privacy as an entity. Their framework compensates the seller for the loss of privacy due to selling of private data. 
A similar approach to pricing in the context of privacy is discussed in~\cite{DBLP:journals/tods/LiLMS14}.

We should finally mention the close connection of query pricing to the measurement of information leakage in programs. In~\cite{kopf2007information}, the authors apply information-theoretic measures, including various entropy measures, to compute the leakage of information from a side-channel attack that attempts to gain access to secret information. \cite{Kopf:2010aa} uses similar ideas to quantify the flow of information in programs, and proposes various approximation techniques to efficiently compute them. 


\section{Conclusion}
\label{sec:conclusion}

In this paper, we explore in depth the design space of arbitrage-free pricing functions. 
We present a characterization of the structure for both answer-dependent and instance-independent
pricing functions, and propose several constructions. 
Our work opens several exciting research questions, including testing which pricing functions 
behave well in practical settings, and exploring the various tradeoffs when deploying a pricing scheme. 

\textbf{Acknowledgements.} We would like to thank Aws Albarghouthi for pointing out the close connection of our work to quantitative information flow and information leakage in side-channel attacks.

\clearpage
\bibliographystyle{abbrv}
\bibliography{ref}  

\end{document}